\pdfoutput=1
\documentclass[10pt,a4paper]{article}
\usepackage[T2A]{fontenc}
\usepackage[utf8]{inputenc}
\usepackage[english]{babel}
\usepackage{amssymb,graphicx}
\usepackage{amsmath,amsfonts}
\usepackage{amsthm}
\usepackage{framed}
\usepackage{xcolor}
\usepackage{fullpage}
\usepackage[makeroom]{cancel}
\usepackage[export]{adjustbox} 
\usepackage{microtype}
\usepackage{hyperref}

 \newcommand{\Tr}{\mathop{\mathrm{Tr}}\nolimits}

 \newtheorem{thm}{Theorem}[section]
 \newtheorem{lem}[thm]{Lemma}
 
 \newtheorem{cor}[thm]{Corollary}
 \newtheorem{prop}[thm]{Proposition}
 \newtheorem{dfn}[thm]{Definition}
 \newtheorem{conj}[thm]{Conjecture}
 
\theoremstyle{remark}
\newtheorem{rem}[thm]{Remark}

\numberwithin{equation}{section}

\title{\textbf{Spiralling branes, affine $qq$-characters and
    elliptic integrable systems}}
\author{Y.~Zenkevich\thanks{yegor.zenkevich@gmail.com}
  \footnote{On leave from ITMP MSU.}\\
  {\small \textit{School of Mathematics, University of Edinburgh,
      UK}}} \date{}
\begin{document}
\maketitle

\begin{abstract}
  We apply the spiralling branes technique introduced
  in~\cite{Zenkevich:2023cza} to many-body integrable systems. We
  start by giving a new $R$-matrix description of the trigonometric
  Ruijsenaars-Schneider (RS) Hamiltonians and eigenfunctions using the
  intertwiners of quantum toroidal algebra. We then consider elliptic
  deformations of the RS system, elucidate how Shiraishi functions
  appear naturally in the process and relate them to certain special
  infinite system of intertwiners of the algebra. We further show that
  there are two distinguished elliptic deformations, one of which
  leads to the conventional elliptic RS Hamiltonians, while the other
  produces trigonometric Koroteev-Shakirov Hamiltonians. Along the way
  we prove the fully noncommutative version of the ``noncommutative
  Jacobi identities'' for affine $qq$-characters recently introduced
  by Grekov and Nekrasov.
\end{abstract}

\section{Introduction}
\label{sec:introduction}
Elliptic integrable systems are a beautiful and fascinating
subject. So are quantum toroidal algebras. In this paper we will try
to bring the two subjects closer together. We will use the simplest
quantum toroidal algebra
$U_{q_1,q_2}(\widehat{\widehat{\mathfrak{gl}}}_1)$~\cite{DIM} which we
call $\mathcal{A}$ henceforth. We will use the following notation for
the deformation parameters of the algebra:
\begin{equation}
  \label{eq:2}
  q_1 = q, \qquad q_2 = t^{-1}, \qquad q_3 = \frac{t}{q}.
\end{equation}
We always assume that $q_i$ are generic, i.e.\ $q_i \neq 0$ and there
are no $n,m \in \mathbb{Z}$ such that $q_1^n q_2^m = 1$. The
definition of $\mathcal{A}$ is given in Appendix~\ref{sec:quant-toro-algebra}.

Representation theory of $\mathcal{A}$ is very rich but we will use
only two types of representations, the Fock representations and the
vector representations. The relevant definitions are collected in
Appendices~\ref{sec:horiz-fock-repr} and~\ref{sec:vect-repr}
respectively.

To explain the setup leading to elliptic integrable systems let us
introduce diagrams corresponding to intertwiners of Fock and vector
representations. The identity operator in a Fock representation is
denoted by a line of the color encoding the type of the Fock
representation:
\begin{gather}
  1|_{\mathcal{F}^{(1,0)}_{q_1, q_2}(u)}  = \quad \includegraphics[valign=c]{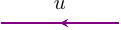},\notag\\
  1|_{\mathcal{F}^{(1,0)}_{q_1, q_3}(u)} = \quad \includegraphics[valign=c]{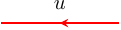},\label{eq:3}\\
  1|_{\mathcal{F}^{(1,0)}_{q_2, q_3}(u)} =
  \quad \includegraphics[valign=c]{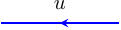}.\notag
\end{gather}
The (absolute value) of the spectral parameter of the representation
$u \in \mathbb{C}^{*}$ should be thought of as the exponentiated
vertical position of the solid line.

Similarly, for vector representation $\mathcal{V}_{q_i}$ we set the
color depending on the $i$ index. Vector representation will appear in
three forms on the pictures (corresponding to slightly different
notions of vector representation as detailed in
Appendix~\ref{sec:horiz-fock-repr}):
\begin{enumerate}
\item No label on a dashed line means an identity operator in
  $\mathcal{V}_{q_i}$:
  \begin{equation}
    \label{eq:138}
   1|_{\mathcal{V}_{q_1}} = \quad \includegraphics[valign=c]{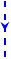}
  \end{equation}

\item A parameter on an intermediate dashed line means a projector on
  the subrepresentation $\mathcal{V}_{q_i}(w)\subset
  \mathcal{V}_{q_i}$.
  \begin{equation}
    \label{eq:139}
    1|_{\mathcal{V}_{q_1}(w)} = \quad \includegraphics[valign=c]{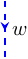}
  \end{equation}
  The (absolute value) of the spectral parameter $w$ should be thought
  of as the exponentiated horizontal position of the dashed line.
      
\item A bra or ket state on the intermediate dashed line means a
  matrix element with or projection on this particular state in
  $\mathcal{V}_{q_i}$.
  \begin{equation}
    \label{eq:140}
        |w\rangle \langle w||_{\mathcal{V}_{q_1}}  = \quad \includegraphics[valign=c]{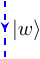}
  \end{equation}
\end{enumerate}
In each of Eqs.~\eqref{eq:138}\eqref{eq:139}\eqref{eq:140} we skip two
more diagrams of other colors representing the other types of vector
representations --- $\mathcal{V}_{q_2}$ and $\mathcal{V}_{q_3}$.

Junctions and crossings of the lines will denote nontrivial
intertwining operators between the corresponding representations.
\begin{enumerate}
\item Junctions between a solid and a dashed line corresponds to a
  vertex operator on the Fock
  representation~\cite{FHHSY},~\cite{Zenkevich:2018fzl}, e.g.\
  \begin{equation}
    \label{eq:146}
   \includegraphics[valign=c]{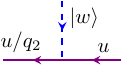}\qquad
   \qquad \text{or} \qquad  \qquad    \includegraphics[valign=c]{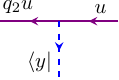}.
  \end{equation}
  Theorems~\ref{prop-phi},~\ref{prop-phi-star}
  provide explicit expressions for such intertwiners.
  
\item Crossings between solid and dashed lines correspond to the
  universal $R$-matrix (see Appendix~\ref{sec:universal-r-matrix-1})
  of the algebra $\mathcal{A}$ evaluated in the representations being
  crossed, e.g.\
  \begin{equation}
    \label{eq:147}
    \includegraphics[valign=c]{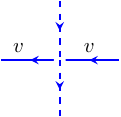} \qquad \qquad
    \text{or} \qquad\qquad \includegraphics[valign=c]{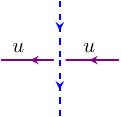}.
  \end{equation}
  In \cite{Haouzi:2024qyo} a particularly simple closed form
  for such $R$-matrices has been found (see
  also~\cite{Gaiotto:2020dsq},~\cite{Zenkevich:2020ufs}). We state
  these results in Theorems~\ref{thm-miura-r},~\ref{thm-r-q} and use
  them throughout the paper.

\end{enumerate}

The crucial ingredient for our representation theory setup is the
notion of compactification. We would like to draw diagrams not only on
the plane but on a cylinder. To do so we introduce the ``portals''
depicted by (multiple) wavy lines. A pair of portals acts as a
teleport: some incoming lines go into one portal and exit through the
other, e.g.
\begin{equation}
  \label{eq:148}
  \includegraphics[valign=c]{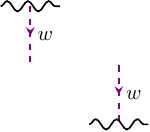}
\end{equation}
In other words one should make identifications on the plane of
the picture along the wavy lines.

Finally one can act on the representations by the grading operator
$p^d \mu^{d_{\perp}}$ (see Eqs.~\eqref{eq:71}--\eqref{eq:72}). We
depict this as follows:
\begin{equation}
  \label{eq:150}
  \includegraphics[valign=c]{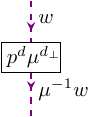}
\end{equation}

It turns out that intertwiner combinations describing the diagrams
which we have introduced above can be used to construct trigonometric
integrable systems such as the trigonometric Ruijsenaars-Schneider
(tRS) system. In fact in sec.~\ref{sec:miura-r-matrix} we show that
both the Hamiltonians and eigenfunctions of the system can be
reproduced in this way. The two parameters of the tRS system usually
denoted $q$ and $t$ in the literature correspond to the parameters
$q_1$ and $q_2^{-1}$ of the algebra $\mathcal{A}$. The spectral
parameters of the vector representations play the role of coordinates
of the tRS particles. Since the spectral parameters live on
$\mathbb{C}^{*}$ the system is trigonometric.

The natural attempt to build elliptic systems using this
correspondence involves compactification of the diagram since one
needs to make periodic identifications on the space $\mathbb{C}^{*}$
on which the spectral parameters of the representations
live. Diagrammatically this means adding portals~\eqref{eq:148} with
shifts~\eqref{eq:150}. The system of lines on a cylinder can then be
either finite (e.g.\ a single line forming a circle) or infinite
(e.g.\ a spiral). The circle diagram in fact produces elliptic
Ruijsenaars-Schneider (eRS) Hamiltonians as proven in
Theorem~\ref{thm-ers}.

Although it seems natural to limit ourselves to the finite
configurations, it turns out that considering infinite spirals is very
fruitful~\cite{Zenkevich:2023cza}. In the current paper we demonstrate
that spiral diagrams are particularly relevant for elliptic integrable
systems. For example, although eRS Hamiltonians are obtained from the
circle diagram as mentioned above, the eigenfunctions cannot be found
in such a simple circular form. Conversely, if elliptic eigenfunctions
are given by a simple circular diagram it turns out that the
Hamiltonians should be extracted from a complicated spiralling
diagrams. We explain this in sec.~\ref{sec:elliptic-rs-system}.

The rest of the paper is organized as follows. In
sec.~\ref{sec:miura-r-matrix} as a warm-up we treat the tRS
system. Using systems of intertwiners of $\mathcal{A}$ we find its
Hamiltonians in sec.~\ref{sec:miura-r-matrix-1} and general
eigenfunctions in sec.~\ref{sec:eigenf-from-intertw}. In
sec.~\ref{sec:defin-shir-funct} we remind the affine screened vertex
operator construction of Shiraishi
functions~\cite{Shir-func,Langmann:2020utd} which are conjectured to
play a role in (non-stationary) eRS model. In
sec.~\ref{sec:affine-vert-oper} we prove that certain spiralling
system of intertwiners reproduces affine screened vertex operators and
the Shiraishi functions. In sec.~\ref{sec:nonc-jacobi-ident} again
using a spiral setup we prove the general version of the
noncommutative Jacobi identity for affine
$qq$-characters~\cite{Grekov:2024ayn}.  After providing some
motivation for the relevance of spiral pictures in
sec.~\ref{sec:ellipt-deform-as} we prove a relation between affine
$qq$-characters and trigonometric limit of Koroteev-Shakirov double
elliptic Hamiltonians in sec.~\ref{sec:korot-shak-syst}. In
sec.~\ref{sec:elliptic-rs-system} we consider the cases when the
spirals degenerate into circles and find that they correspond to eRS
Hamiltonians and the limit of Shiraishi functions considered
in~\cite{Fukuda:2020czf}.

\section{Trigonometric Ruijsenaars-Schneider system revisited}
\label{sec:miura-r-matrix}
In this section we give a new $R$-matrix formulation of the tRS
system. The $R$-matrices featuring in it are those of the quantum
toroidal algebra $\mathcal{A}$ taken in tensor product of Fock and
vector representations. In this case the $R$-matrices turn out to have
an especially simple explicit form. This allows us to build both the
Hamiltonians (sec.~\ref{sec:miura-r-matrix-1}) and the eigenfunctions
(sec.~\ref{sec:eigenf-from-intertw}) of the system in a systematic
way. Similar construction of eigenfunctions has appeared
in~\cite{Nedelin:2017nsb},~\cite{Zenkevich:2018fzl} but here we
provide a streamlined derivation which naturally incorporates all
higher tRS Hamiltonians.

\subsection{tRS Hamiltonians from Miura $R$-matrix}
\label{sec:miura-r-matrix-1}
\begin{thm}[\cite{Haouzi:2024qyo}]
  The universal $R$-matrix of $\mathcal{A}$ evaluated on
  $\mathcal{V}_{q_1}^{*} \otimes \mathcal{F}^{(1,0)}_{q_2, q_3}(v)$ is
  given by
\begin{equation}
  \label{eq:5}
  R_{q_2,q_3}^{q_1}(x|v) \stackrel{\mathrm{def}}{=}  \check{\mathcal{R}}|_{\mathcal{V}^{*}_{q_1} \otimes
    \mathcal{F}_{q_2,q_3}^{(1,0)}(v)}  = \quad \includegraphics[valign=c]{figures/d3-ns5-diff-crop} \quad=   V^{+}_{q_2,q_3}(x) q_1^{\frac{1}{2} x \partial_x} -
  \left(  v/q_1\right)^{-1} V^{-}_{q_2,q_3}(x) q_1^{-\frac{1}{2} x \partial_x} 
\end{equation}
where $\mathcal{V}_{q_1}^{*}$ is the representation on the space of
functions $\mathbb{C}((x))$ defined in Appendix~\ref{sec:vect-repr},
\begin{equation}
  \label{eq:43}
  V^{\pm}_{q_i,q_j}(x) = \exp \left[ \sum_{n \geq 1} \frac{x^{\mp n}}{n}
    (1-q_3^{\mp n}) a^{(q_i, q_j)}_{\pm n} \right], \qquad i \neq j
  \in \{1,2,3\}
\end{equation}
and $a_n^{(q_i,q_j)}$ are Heisenberg generators acting on
$\mathcal{F}_{q_2, q_3}^{(1,0)}(v)$ satisfying the
relations~\eqref{eq:44}.
\label{thm-miura-r}
\end{thm}
\begin{rem}
  We have used the version of the vector representation
  $\mathcal{V}_{q_1}^{*}$ that is dual to the representation
  conventionally used in the literature (see
  e.g.~\cite{Zenkevich:2018fzl, Zenkevich:2020ufs}). In our convention
  this means that the operator $R(x|v)$ acts on a wavefunction that
  corresponds to a diagram sitting \emph{above} the crossing in the
  picture.
\end{rem}

Using Theorem~\ref{thm-miura-r} it is easy to find the $R$-matrix
acting on $(\mathcal{V}_{q_1}^{*})^{\otimes N} \otimes
\mathcal{F}_{q_2,q_3}^{(1,0)}(v)$: it is given by the composition of
$N$ $R$-matrices~\eqref{eq:5} along the horizontal line. The result is
a difference operator in $x_i$ $i=1,\ldots,N$ which is also a vertex
operator acting on $\mathcal{F}_{q_2,q_3}^{(1,0)}(v)$.
\begin{dfn}
  Let $\mathcal{O}(\vec{x}|v): (\mathcal{V}_{q_1}^{*})^{\otimes N} \to
  (\mathcal{V}_{q_1}^{*})^{\otimes N}$ be matrix element of the
  $R$-matrix acting on $(\mathcal{V}_{q_1}^{*})^{\otimes N} \otimes
  \mathcal{F}_{q_2,q_3}^{(1,0)}(v)$ between the vacuum states in the
  Fock representation:
  \begin{multline}
    \label{eq:46}
    \mathcal{O}_N(\vec{x}|v) = \langle \varnothing,
    v|\check{\mathcal{R}}|_{(\mathcal{V}^{*}_{q_1})^{\otimes N} \otimes
      \mathcal{F}_{q_2,q_3}^{(1,0)}(v)}| \varnothing, v\rangle =\\
    =\langle \varnothing, v| \prod_{i=1}^{
      \begin{smallmatrix}
        \to\\
        N
      \end{smallmatrix}
    } R_{q_2,q_3}^{q_1}(x_i|v) | \varnothing, v\rangle = \quad \includegraphics[valign=c]{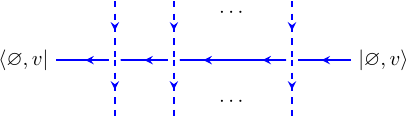}
  \end{multline}
  where $x_i$ is understood as the variable on which the functions in
  the $i$-th $\mathcal{V}_{q_1}^{*}$ tensor factor depend. We call
  $\mathcal{O}(\vec{x}|v)$ the transfer matrix.
\end{dfn}
\begin{dfn}
  Trigonometric Ruijsenaars-Schneider Hamiltonians\footnote{These are
    also known as Macdonald difference operators.} are the following
  difference operators
  \begin{equation}
    \label{eq:47}
    H_k^{\mathrm{tRS},(N)}(\vec{x},q_1,q_2) = q_2^{-\frac{k(k-1)}{2}}\sum_{
      \begin{smallmatrix}
        I \subseteq \{1,\ldots, N \}\\
        |I|=k
      \end{smallmatrix}
    } \prod_{i \in I} \prod_{j \not \in I} \frac{q_2^{-1} x_i -
      x_j}{x_i - x_j} \prod_{i \in I} q_1^{x_i \partial_{x_i}}, \qquad k=0,\ldots,N.
  \end{equation}
\end{dfn}
Trigonometric RS Hamiltonians commute~\cite{Macdonald} and their
common polynomial eigenfunctions are given by Macdonald symmetric
polynomials $P_{\lambda}(\vec{x})$:
\begin{equation}
  \label{eq:48}
  H_k^{\mathrm{tRS},(N)}(\vec{x},q_1,q_2)
  P^{(q_1,q_2)}_{\lambda}(\vec{x}) = e_k(\vec{u}) P^{(q_1,q_2)}_{\lambda}(\vec{x}),
\end{equation}
where $u_i = q_1^{\lambda_i} q_2^{i-N}$ for $i=1,\ldots,N$, $e_k$ are
elementary symmetric functions and $\lambda = (\lambda_1, \lambda_2,
\ldots, \lambda_N)$ is a partition (Young diagram). In what follows we
will mostly be interested in the eigenfunctions for generic parameters
$u_i$; these are non-polynomial.

We notice that $\mathcal{O}_N(\vec{x}|v)$ is the generating function of
tRS Hamiltonians:
\begin{thm}
\begin{equation}
  \label{eq:45}
  (F_{q_1,q_2}(\vec{x}))^{-1} \mathcal{O}_N(\vec{x}|v)
  F_{q_1,q_2}(\vec{x}) = q_2^{\frac{N(N-1)}{2}} \sum_{k=0}^N
  (-v/q_1)^{k-N}  H^{\mathrm{tRS},(N)}_k(\vec{x},q_1,q_2) q_1^{-\frac{1}{2} \sum_{i=1}^N x_i \partial_{x_i}},
\end{equation}
where
\begin{equation}
  \label{eq:49}
  F_{q_1,q_2}(\vec{x}) = \prod_{k<l}^N \left[ x_k^{- \frac{\ln
        q_2}{\ln q_1}} \frac{\left( \frac{x_l}{x_k} ; q_1
      \right)_{\infty}}{\left( q_2^{-1} \frac{x_l}{x_k} ; q_1 \right)_{\infty}} \right].
\end{equation}
\label{thm-tRS-ham}
\end{thm}
\begin{rem}
  The overall shift $q_1^{-\frac{1}{2} \sum_{i=1}^N
    x_i \partial_{x_i}}$ commutes with all
  $H_k^{\mathrm{tRS},(N)}(\vec{x},q_1,q_2)$.
\end{rem}
\begin{proof}
  Let us first evaluate $\mathcal{O}_N(\vec{x}|v)$. Expanding the
  product we get
  \begin{multline}
    \label{eq:52}
    \mathcal{O}_N(\vec{x}|v) = \langle \varnothing, v|\prod_{i=1}^{
      \begin{smallmatrix}
        \to\\
        N
      \end{smallmatrix}
    } \left( V^{+}_{q_2,q_3}(x_i) q_1^{\frac{1}{2} x_i \partial_{x_i}} - \left(
        v/q_1\right)^{-1} V^{-}_{q_2,q_3}(x_i) q_1^{-\frac{1}{2}
        x_i \partial_{x_i}} \right) | \varnothing, v\rangle =\\
    = \sum_{k=0}^N \sum_{ 1 \leq i_1 < \ldots < i_k \leq N} \left( -
      v/q_1\right)^{k-N} \langle \varnothing, v| V^{-}_{q_2,q_3}(x_1) \cdots
    V^{-}_{q_2,q_3}(x_{i_1 - 1}) V^{+}_{q_2,q_3}(x_{i_1}) V^{-}_{q_2,q_3}(x_{i_1 - 1}) \cdots\\
    \cdots V^{-}_{q_2,q_3}(x_{i_2 - 1}) V^{+}_{q_2,q_3}(x_{i_2})
    V^{-}_{q_2,q_3}(x_{i_2 - 1})\cdots V^{-}_{q_2,q_3}(x_{i_k - 1})
    V^{+}_{q_2,q_3}(x_{i_k}) V^{-}_{q_2,q_3}(x_{i_k - 1})\cdots\\
    \cdots
    V^{-}_{q_2,q_3}(x_N) | \varnothing, v\rangle q_1^{\sum_{j=1}^k
      x_{i_j} \partial_{x_{i_j}} - \frac{1}{2} \sum_{i=1}^N
      x_i \partial_{x_i}}.
  \end{multline}
  Next we use the commutation relation between $V^{+}_{q_2,q_3}(x)$ and
  $V^{-}_{q_2,q_3}(y)$ to move all $V^{+}_{q_2,q_3}$ to the right of all $V^{-}_{q_2,q_3}$ so that
  eventually both act trivially on the vacuum vectors of
  $\mathcal{F}_{q_2,q_3}^{(1,0)}(v)$:
  \begin{align}
    \label{eq:53}
    V^{+}_{q_2,q_3}(z) |\varnothing, v\rangle &= |\varnothing, v\rangle,\\
    \langle\varnothing, v| V^{-}_{q_2,q_3}(z) &= \langle\varnothing, v|.
  \end{align}
  We have
  \begin{equation}
    \label{eq:50}
    V^{+}_{q_i,q_j}(x)V^{-}_{q_i,q_j}(y) = \psi_{q_i,q_j} \left(
      \frac{y}{x} \right)V^{-}_{q_i,q_j}(y) V^{+}_{q_i,q_j}(x), , \qquad i \neq j \in \{1,2,3\},
  \end{equation}
  where
  \begin{equation}
    \label{eq:51}
    \psi_{q_i,q_j}(z) = \frac{(1 - q_i z)(1 - q_j z)}{(1 - z)(1 - q_i
      q_j z)}, \qquad i \neq j \in \{1,2,3\}.
  \end{equation}
  Plugging commutation relation~\eqref{eq:50} into Eq.~\eqref{eq:52}
  we get
  \begin{equation}
    \label{eq:54}
    \mathcal{O}_N(\vec{x}|v) = \sum_{k=0}^N  \left( -
      v/q_1\right)^{k-N} \sum_{
      \begin{smallmatrix}
        I \subseteq \{1,\ldots,N\}\\
        |I| = k
      \end{smallmatrix}
} \prod_{i \in I} \prod_{
      \begin{smallmatrix}
        j \not \in I\\
        j>i
      \end{smallmatrix}
} \psi_{q_2,q_3} \left( \frac{x_j}{x_i} \right)  \prod_{i\in I} q_1^{
      x_i \partial_{x_i}} q_1^{- \frac{1}{2} \sum_{j=1}^N
      x_j \partial_{x_j}}
  \end{equation}
  Finally we perform the conjugation with the function
  $F_{q_1,q_2}(\vec{x})$. An explicit computation gives
  \begin{multline}
    \label{eq:55}
    (F_{q_1,q_2}(\vec{x}))^{-1} \prod_{i\in I} q_1^{
      x_i \partial_{x_i}} q_1^{- \frac{1}{2} \sum_{j=1}^N
      x_j \partial_{x_j}}
    F_{q_1,q_2}(\vec{x})=\\
    = q_2^{\frac{N(N-1)}{2}} \prod_{i \in I} \left[
 \prod_{ \begin{smallmatrix}
        j \not \in I\\
        j<i
      \end{smallmatrix}} \frac{1 - q_2^{-1} \frac{x_i}{x_j}}{1 - \frac{x_i}{x_j}}
    \prod_{ \begin{smallmatrix}
        j \not \in I\\
        j>i
      \end{smallmatrix}} q_2^{-1} \frac{1 - q_1^{-1}
      \frac{x_i}{x_j}}{1 - (q_1 q_2)^{-1} \frac{x_i}{x_j}} \prod_{ \begin{smallmatrix}
        j \in I\\
        j<i
      \end{smallmatrix}} \frac{1 - (q_1 q_2)^{-1} \frac{x_j}{x_i}}{1 -
      q_1^{-1} \frac{x_j}{x_i}} \prod_{ \begin{smallmatrix}
        j \in I\\
        j>i
      \end{smallmatrix}} q_2^{-1} \frac{1 - \frac{x_j}{x_i}}{1 -
      q_2^{-1} \frac{x_j}{x_i}} \right]\times\\
  \times \prod_{i\in I} q_1^{
    x_i \partial_{x_i}} q_1^{- \frac{1}{2} \sum_{j=1}^N
    x_j \partial_{x_j}}
  \end{multline}
  Substituting Eq.~\eqref{eq:55} into Eq.~\eqref{eq:54} and cancelling
  the factors we obtain the statement of the theorem.
\end{proof}

It is instructive to see how the commutativity of
$H_k^{\mathrm{tRS},(N)}$ follows from the Yang-Baxter equation for the
universal $R$-matrix of the algebra $\mathcal{A}$. To see this we will
need the action of the universal $R$-matrix on the tensor product of
vacuum vectors in a pair of Fock representations.
\begin{prop}[\cite{FJMM-R}]
  Let $\check{\mathcal{R}}$ be the universal $R$-matrix of
  $\mathcal{A}$ (see Appendix~\ref{sec:universal-r-matrix-1}), then
  \begin{align}
    \label{eq:56}
    \check{\mathcal{R}}|_{\mathcal{F}_{q_2,q_3}^{(1,0)}(u) \otimes
      \mathcal{F}_{q_2,q_3}^{(1,0)}(v)} | \varnothing, u\rangle
  \otimes | \varnothing, v\rangle &= f_{q_2,q_3}^{q_2,q_3} \left( \frac{u}{v}  \right)| \varnothing, v\rangle
  \otimes | \varnothing, u\rangle,\\
      \check{\mathcal{R}}|_{\mathcal{F}_{q_1,q_2}^{(1,0)}(u) \otimes
      \mathcal{F}_{q_2,q_3}^{(1,0)}(v)} | \varnothing, u\rangle
  \otimes | \varnothing, v\rangle &= f_{q_2,q_3}^{q_1,q_2} \left( \frac{u}{v}  \right) | \varnothing, v\rangle
  \otimes | \varnothing, u\rangle.\label{eq:58}
  \end{align}
  where
  \begin{align}
    \label{eq:60}
   f_{q_2,q_3}^{q_2,q_3} (z) &=\exp \left[ - \sum_{n \geq 1}
    \frac{(1 - q_1^{-n})}{(1-q_2^n)(1-q_3^n)} z^n \right],\\
 f_{q_2,q_3}^{q_1,q_2}(z) &= \exp \left[ \sum_{n \geq 1}
    \frac{q_1^{-n}}{(1-q_2^n)} z^n \right].\label{eq:61}
\end{align}
\label{prop-R-hor-vac}
\end{prop}
\begin{rem}
  The $R$-matrix for a tensor product of a pair of representations (or
  indeed any intertwining operator) is defined up to an arbitrary
  rescaling by a function of the parameters of the representations on
  which it acts. The universal $R$-matrix~\eqref{eq:4} provides a
  choice of normalization for every pair of representations. 
\end{rem}
\begin{thm}
  \begin{equation}
    \label{eq:57}
    \mathcal{O}_N(\vec{x}|u) \mathcal{O}_N(\vec{x}|v) =
    \mathcal{O}_N(\vec{x}|v) \mathcal{O}_N(\vec{x}|u).
  \end{equation}
  and therefore
  \begin{equation}
    \label{eq:59}
    [H_k^{\mathrm{tRS},(N)}(\vec{x},q_1,q_2),
    H_l^{\mathrm{tRS},(N)}(\vec{x},q_1,q_2)] = 0 \qquad k,l = 1,\ldots,N.
  \end{equation}
  \label{thm-O-comm}
\end{thm}
\begin{proof}
  Standard manipulations with $R$-matrices using the Yang-Baxter
  equation are easier to draw than to write out. We start by drawing
  the product of $\mathcal{O}_N(\vec{x}|u)$ for two different values of
  spectral parameters:
  \begin{equation}
    \mathcal{O}_N(\vec{x}|u) \mathcal{O}_N(\vec{x}|v)
    =\quad \includegraphics[valign=c]{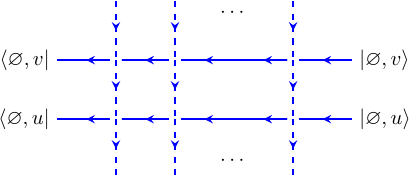}
  \end{equation}
  We insert a product $1 = \check{\mathcal{R}}
  \check{\mathcal{R}}^{-1}$ acting on
  $\mathcal{F}_{q_2,q_3}^{(1,0)}(u) \otimes
  \mathcal{F}_{q_2,q_3}^{(1,0)}(v)$ and then use the Yang-Baxter
  equation repeatedly to move $\check{\mathcal{R}}$ to the very left
  of the picture and $\check{\mathcal{R}}^{-1}$ to the very right. The
  result is :
  \begin{multline}
    \mathcal{O}_N(\vec{x}|u) \mathcal{O}_N(\vec{x}|v)  =\quad \includegraphics[valign=c]{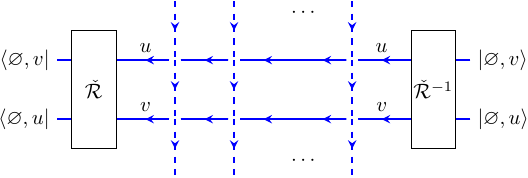}
    \quad =\\
    =\quad f_{q_2,q_3}^{q_2,q_3} \left( \frac{v}{u} \right)
    \times \includegraphics[valign=c]{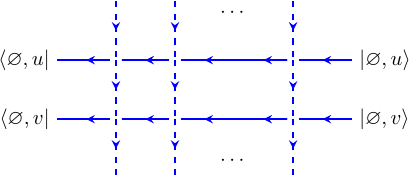} \times
    \left( f_{q_2,q_3}^{q_2,q_3} \left( \frac{v}{u} \right)
    \right)^{-1} \quad = \mathcal{O}_N(\vec{x}|v)
    \mathcal{O}_N(\vec{x}|u),
    \label{eq:62}
  \end{multline}
  In the second line of Eq.~\eqref{eq:62} we have used
  Proposition~\ref{prop-R-hor-vac}. Since the spectral parameters do
  not change along the horizontal lines in this case, the resulting
  function $f_{q_2,q_3}^{q_2,q_3} (u/v)$ and its inverse cancel giving
  the commutativity of $\mathcal{O}_N(\vec{x}|v)$.
\end{proof}
In the next section we use similar $R$-matrix manipulations to obtain
the eigenfunctions of the operators $\mathcal{O}_N(\vec{x}|v)$.
 
\subsection{Trigonometric RS eigenfunctions from intertwiners}
\label{sec:eigenf-from-intertw}
Let us recall some more results about the intertwiners and
$R$-matrices of $\mathcal{A}$.
\begin{prop}[\cite{FHHSY}]
  For any $u \in \mathbb{C}^*$ there exists an
  intertwining operator $\Phi^{q_1}_{q_1,q_2}: \mathcal{V}_{q_1}^{*} \otimes
  \mathcal{F}_{q_1,q_2}^{(1,0)}(u) \to
  \mathcal{F}_{q_1,q_2}^{(1,0)}(u/q_2)$ which is unique up to
  rescaling and given by
\begin{multline}
  \label{eq:21}
  \Phi^{q_1}_{q_1,q_2}(w) \stackrel{\mathrm{def}}{=} \Phi^{q_1}_{q_1,q_2} (|w\rangle \otimes
  \ldots) =
  \quad \includegraphics[valign=c]{figures/vect-vert-color-simpl-crop}
  \quad =\\
  = q_2^{-Q} w^{\frac{P}{\ln q_1}} \exp \left[ -
    \sum_{n\geq 1} \frac{w^n}{n} \frac{1-q_2^n}{1-q_1^n} a_{-n}
  \right] \exp \left[ \sum_{n\geq 1} \frac{w^{-n}}{n}
    \frac{1-q_2^{-n}}{1-q_1^{-n}} a_n \right]
\end{multline}
where $a_n$ are the Heisenberg generators acting on
$\mathcal{F}_{q_1,q_2}^{(1,0)}(u)$ satisfying
\begin{equation}
  \label{eq:63}
  [a_n, a_m] = n \frac{1 - q_1^{|n|}}{1 - q_2^{-|n|}} \delta_{n+m,0},
\end{equation}
and $P$ and $Q$ are the zero modes, $P$ having eigenvalue $\ln u$ on
$\mathcal{F}_{q_1,q_2}^{(1,0)}(u)$ (see
Appendix~\ref{sec:horiz-fock-repr}).
\label{prop-phi}
\end{prop}
\begin{rem}
  The intertwining operator~\eqref{eq:21} is defined up to rescaling
  not only by complex numbers but by arbitrary functions of $P$ (or
  equivalently of the spectral parameter $u$) and $q_1$-periodic
  functions of $w$.
\end{rem}
\begin{prop}
  The dual intertwiner $\Phi_{q_1,q_2}^{*q_1}:
  \mathcal{F}_{q_1,q_2}^{(1,0)}(u) \to \mathcal{F}_{q_1,
    q_2}^{(1,0)}(q_2 u) \otimes \mathcal{V}^{*}_{q_1}$ is given
  by~\cite{Zenkevich:2018fzl}:
\begin{multline}
  \label{eq:23}
  \Phi^{*q_1}_{q_1,q_2}(y) \stackrel{\mathrm{def}}{=}  (\ldots \otimes \langle y| ) \Phi^{*q_1}_{q_1,q_2} =
  \quad  \includegraphics[valign=c]{figures/vect-vert-dual-color-crop}
  \quad =\\
  = q_2^Q y^{-\frac{\ln (q_2) + P}{\ln q_1}} \exp \left[ \sum_{n\geq 1} \frac{y^n}{n}
    q_3^{\frac{n}{2}} \frac{1-q_2^n}{1-q_1^n}
    a_{-n} \right] \exp \left[ - \sum_{n\geq 1} \frac{y^{-n}}{n}
    q_3^{\frac{n}{2}} \frac{1-q_2^{-n}}{1-q_1^{-n}}
    a_n \right],
\end{multline}
where $a_n$, $P$ and $Q$ are as in Proposition~\ref{prop-phi}.
\label{prop-phi-star}
\end{prop}

In what follows we will also need an expression for the $R$-matrix
acting on $\mathcal{V}_{q_1}^{*} \otimes
\mathcal{F}_{q,t^{-1}}^{(1,0)}(u)$.

\begin{thm}[\cite{Haouzi:2024qyo}]
  The only nontrivial matrix elements of the $R$-matrix of the algebra
  $\mathcal{A}$ acting on $\mathcal{V}_{q_1} \otimes
  \mathcal{F}_{q_1,q_2}^{(1,0)}(u)$ are
  \begin{multline}
    \label{eq:31}
    (\ldots) \otimes \left\langle \left.  w
        \left|\check{\mathcal{R}}|_{\mathcal{V}_{q_1} \otimes
            \mathcal{F}_{q_1,q_2}^{(1,0)}(u)}\right|
        q_3^{\frac{1}{2}}\, q_1^{-n} w \right. \right\rangle \otimes
    (\ldots)=\quad\includegraphics[valign=c]{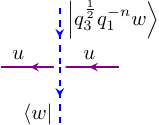}\quad
    =   \\
    = q_3^{-\frac{\ln (q_2) + P}{2 \ln q_1}} \frac{\left( q_3^{-1}
        ;q_1\right)_{\infty}}{(q_1;q_1)_{\infty}}
    \Phi^{q_1}_{q_1,q_2}(q_3^{\frac{1}{2}}\,q_1^{-n} w)
    \Phi^{*q_1}_{q_1,q_2}\left( w \right) = q_3^{-\frac{\ln (q_2) +
        P}{2 \ln q_1}} \frac{\left( q_3^{-1}
        ;q_1\right)_{\infty}}{(q_1;q_1)_{\infty}}
    \times\quad \includegraphics[valign=c]{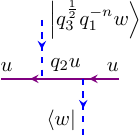}
  \end{multline}
  for $n \in \mathbb{Z}_{\geq 0}$.
  \label{thm-r-q}
\end{thm}
\begin{cor}[\cite{Haouzi:2024qyo}]
  The $R$-matrix acting on $\mathcal{V}^{*}_{q_1} \otimes
  \mathcal{F}_{q_1,q_2}^{(1,0)}(u)$ is given by the following
  $q_1$-difference operator on functions of $x$:
  \begin{multline}
    \label{eq:1}
    R_{q_1,q_2}^{q_1}(x|u) \stackrel{\mathrm{def}}{=}
    \check{\mathcal{R}}|_{\mathcal{V}^{*}_{q_1} \otimes
      \mathcal{F}_{q_1,q_2}^{(1,0)}(u)}=\\
    =
    \quad\includegraphics[valign=c]{figures/d3-ns5-same-inv-pure-crop}\quad
    = q_3^{-\frac{\ln (q_2) + P}{2 \ln q_1}} \frac{\left( q_3^{-1}
        ;q_1\right)_{\infty}}{(q_1;q_1)_{\infty}} \sum_{n \geq 0}
    \Phi^{q_1}_{q_1,q_2}(q_3^{\frac{1}{2}}\,q_1^{-n} x)
    \Phi^{*q_1}_{q_1,q_2}\left( x \right) q_3^{\frac{1}{2}
      x \partial_x} q_1^{-nx\partial_{x}}.
  \end{multline}
\end{cor}
\begin{rem}
  In fact, performing the normal ordering of $\Phi$ and $\Phi^{*}$ in
  the r.h.s.\ of Eq.~\eqref{eq:1} one finds that
  \begin{equation}
    \label{eq:22}
    \Phi^{q_1}_{q_1,q_2}(q_3^{\frac{1}{2}}\,q_1^{-n} x) \Phi^{*q_1}_{q_1,q_2}\left(
      x \right) = 0 \qquad \text{for } n \in \mathbb{Z}_{<0}.
  \end{equation}
  Therefore we can formally extend the summation in Eq.~\eqref{eq:1}
  to all $n \in \mathbb{Z}$.
\end{rem}
\begin{rem}
  The difference operator~\eqref{eq:1} in our convention is understood
  to act on a wavefunction corresponding to a piece of the diagram
  sitting above the picture.
\end{rem}

We give an explicit combination of intertwining operators and
$R$-matrices of the algebra $\mathcal{A}$ that produces the joint
eigenfunction for the tRS Hamiltonians with generic eigenvalues.
\begin{dfn}
  Let $u_i \in \mathbb{C}^{*}$, $i=1,\ldots,N$ be generic
  parameters. Define $f_N (\vec{x}|\vec{u}|q_1,q_2)$ to be the
  following matrix element of intertwiners:
\begin{equation}
  \label{eq:19}
   f_N
  (\vec{x}|\vec{u}|q_1,q_2) = (G_{q_1,q_2}(\vec{x}))^{-1}\times \quad\includegraphics[valign=c]{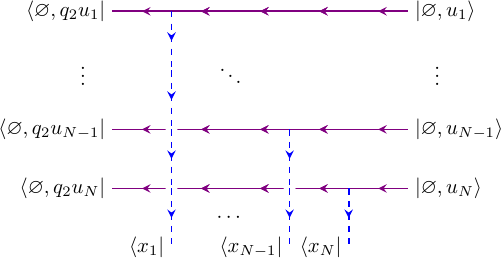}
\end{equation}
where
\begin{equation}
  \label{eq:74}
  G_{q_1,q_2}(\vec{x}) = \prod_{i<j} \frac{\left(
      \frac{x_j}{x_i} ; q_1 \right)_{\infty} \left(
      q_1 \frac{x_j}{x_i} ; q_1 \right)_{\infty}}{\left(
      \frac{1}{q_2} \frac{x_j}{x_i} ; q_1 \right)_{\infty} \left(
      \frac{1}{q_3} \frac{x_j}{x_i} ; q_1 \right)_{\infty}}.
\end{equation}
\end{dfn}
\begin{thm} Let
  \begin{equation}
    \label{eq:36}
    \psi_N (\vec{x}|\vec{u}|q_1,q_2) =   \prod_{i<j} \left[
      x_i^{\frac{\ln q_2}{\ln q_1}} \frac{\left(
      q_1 \frac{x_j}{x_i} ; q_1 \right)_{\infty}}{\left(
      \frac{1}{q_3} \frac{x_j}{x_i} ; q_1 \right)_{\infty}}  \right] f_N (\vec{x}|\vec{u}|q_1,q_2). 
  \end{equation}
  Then
  \begin{equation}
    \label{eq:32}
    H_k^{\mathrm{tRS},(N)} \psi_N (\vec{x}|\vec{u}|q_1,q_2) = e_k(\vec{u}) \psi_N (\vec{x}|\vec{u}|q_1,q_2),
  \end{equation}
  where $e_k$ are elementary symmetric functions.
\end{thm}
\begin{proof}
  Consider the difference operator $\mathcal{O}_N(\vec{x}|v)$ defined in
  Eq.~\eqref{eq:46}. Let it act on the r.h.s.\ of
  Eq.~\eqref{eq:19}. We get the following diagram:
  \begin{equation}
    \label{eq:65}
    \mathcal{O}_N(\vec{x}|v) G_{q_1,q_2}(\vec{x})
    f^{\mathrm{tRS}}_N (\vec{x}|\vec{u}|q_1,q_2)
    =\quad\includegraphics[valign=c]{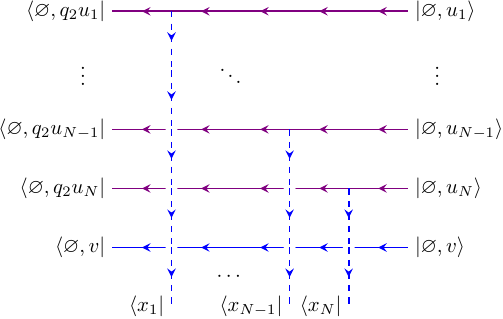}
  \end{equation}
  Similarly to the proof of Theorem~\ref{thm-O-comm} we insert the
  $R$-matrix and its inverse acting on
  $\mathcal{F}^{(1,0)}_{q_1,q_2}(u_N) \otimes
  \mathcal{F}^{(1,0)}_{q_2,q_3}(v)$ and move them to the left and
  right of the diagram respectively. The blue and lowest violet lines
  are exchanged in the process. After that we use
  Proposition~\ref{prop-R-hor-vac} and get a nontrivial prefactor
  depending on $\frac{u_N}{v}$ since the spectral parameter on
  $\mathcal{F}^{(1,0)}_{q_1,q_2}$ changes from to the right of the
  picture $u_N$ to $q_2 u_N$ to the left of it.
  \begin{multline}
    \label{eq:39}
    \mathcal{O}_N(\vec{x}|v) G_{q_1,q_2}(\vec{x})
    f^{\mathrm{tRS}}_N (\vec{x}|\vec{u}|q_1,q_2)=\\
    \\=\quad\includegraphics[valign=c]{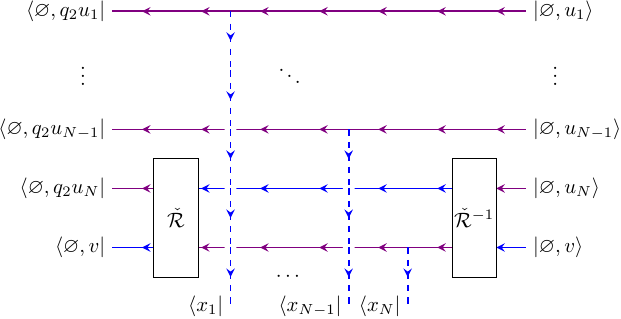}=\\
    \\=\quad f_{q_2,q_3}^{q_1,q_2} \left( \frac{q_2 u}{v} \right)
    \times \includegraphics[valign=c]{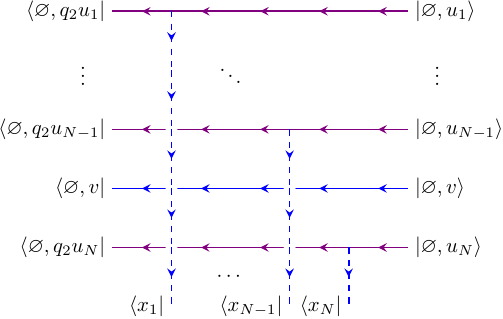}
    \times \left( f_{q_2,q_3}^{q_1,q_2} \left( \frac{u}{v} \right)
    \right)^{-1}
  \end{multline}
  Repeating the previous steps for all the other violet lines above
  the blue line we get:
  \begin{multline}
  \label{eq:64}
  \mathcal{O}_N(\vec{x}|v)G_{q_1,q_2}(\vec{x}) f^{\mathrm{tRS}}_N (\vec{x}|\vec{u}|q_1,q_2)
  =\\
  \\
  = \prod_{i=1}^N f_{q_2,q_3}^{q_1,q_2} \left( \frac{q_2 u_i}{v}
  \right)
  \times \includegraphics[valign=c]{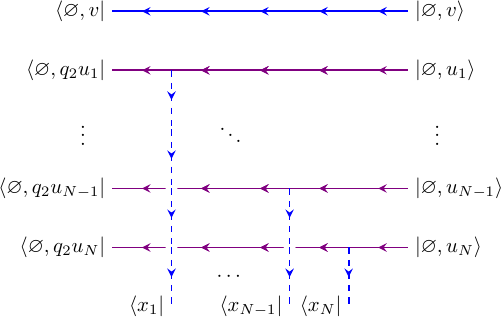}
  \times \prod_{i=1}^N \left( f_{q_2,q_3}^{q_1,q_2} \left(
      \frac{u_i}{v} \right)
  \right)^{-1}.
\end{multline}
Using the explicit form of the function $f_{q_2,q_3}^{q_1,q_2}(z)$ we
eventually find
  \begin{equation}
    \mathcal{O}_N(\vec{x}|v) G_{q_1,q_2}(\vec{x}) f^{\mathrm{tRS}}_N (\vec{x}|\vec{u}|q_1,q_2)
    = \prod_{i=1}^N \left( 1 - \frac{1}{q_1} \frac{u_i}{v} \right)
    G_{q_1,q_2}(\vec{x}) f^{\mathrm{tRS}}_N (\vec{x}|\vec{u}|q_1,q_2).
\end{equation}
Dividing both sides of Eq.~\eqref{eq:64} by $F_{q_1,q_2}(\vec{x})$,
using Theorem~\ref{thm-tRS-ham} and expanding in negative powers of
$v$ we obtain Eq.~\eqref{eq:32}.
\end{proof}

\section{Shiraishi functions from spiralling branes}
\label{sec:shir-funct-their}
In this section we give a proof of the theorem stated
in~\cite{Zenkevich:2023cza} that relates the Shiraishi
wavefunctions~\cite{Shir-func}, \cite{Langmann:2020utd} to certain
infinite systems of intertwiners of the quantum toroidal algebra
$\mathcal{A}$. The system of intertwiners has the form of an infinite
spiral winding around a cylinder.

We start by recalling the construction of the Shiraishi functions from
the system of screened affine vertex operators~\cite{Shir-func} in
sec.~\ref{sec:defin-shir-funct}. In sec.~\ref{sec:affine-vert-oper} we
reinterpret all the elements in the Shiraishi's construction in terms
of intertwiners of representations of the quantum toroidal algebra
$\mathcal{A}$.  To make the identification with~\cite{Shir-func},
\cite{Langmann:2020utd} easier in this section we use the parameters
$q = q_1$ and $t = q_2^{-1}$.

\subsection{The definition of Shiraishi functions}
\label{sec:defin-shir-funct}
The construction is based on the system of vertex operators and
screening currents corresponding to the affine root system
$\widehat{A}_{N-1}$. In this and the following section we assume $N
 \geq 3$.
\begin{dfn}
  Let $\kappa$, $q$ and $t$ be generic parameters. Let the ``weight
  type'' Heisenberg generators $\beta^i_n$, $i=0,\ldots,N-1$, $n \in
  \mathbb{Z} \backslash \{0\}$ satisfy the following commutation
  relations:
  \begin{align} [\beta^i_n, \beta^i_m] &= n \frac{1- t^n}{1-q^n} \frac{1
    - \left( \kappa \frac{q}{t} \right)^n}{1-\kappa^n} \delta_{n+m,0},
  \qquad i
  = 0,\ldots, N-1,\notag\\
  [\beta^i_n, \beta^j_m] &= n \frac{1- t^n}{1-q^n} \frac{1 - \left(
      \frac{q}{t} \right)^n}{1-\kappa^n} \kappa^{n
    \frac{j-i}{N}}\delta_{n+m,0}, \qquad i <j,  \label{eq:8}\\
  [\beta^i_n, \beta^j_m] &= n \kappa^n \frac{1- t^n}{1-q^n} \frac{1 -
    \left( \frac{q}{t} \right)^n}{1-\kappa^n} \kappa^{n \frac{j-i}{N}}
  \delta_{n+m,0}, \qquad i >j.\notag
\end{align}
Let $|0\rangle$ be the vacuum vector such that $\beta^i_n |0\rangle =
0$ for all $i = 0, \ldots, N-1$, $n > 0$ and let $\mathfrak{F}_N$ be
the vector space spanned by $\beta^{i_1}_{-n_1} \cdots
\beta^{i_k}_{-n_k} |0\rangle$ for $k \in \mathbb{Z}_{\geq 0}$, $i_j =
0,\ldots, N-1$, $n >0$.
\end{dfn}
Having the ``weight type'' generators $\beta^i_n$ it is natural to
define the (unscreened) affine vertex operators.
\begin{dfn}
  We call
  \begin{equation}
    \label{eq:10}
    \phi_i(z) = :\exp \left( \sum_{n \neq 0} \frac{1}{n} \beta^i_n
      z^{-n} \right):, \qquad i = 0,\ldots, N-1,
  \end{equation}
  affine vertex operators, where $:(\ldots):$ stands for the normal
  ordered product in which $\beta^i_n$ with positive $n$ are to the
  right of those with negative $n$.
\end{dfn}
We also introduces the differences of the ``weight type'' Heisenberg
generators.
\begin{dfn}
  Let the ``root type'' Heisenberg generators $\alpha^i_n$ be given by
  \begin{equation}
    \label{eq:9}
    \alpha^i_n =
    \begin{cases}
      \kappa^{-\frac{n}{N}} \beta^{N-1}_n - \beta^0_n, & i = 0,\\
      \kappa^{-\frac{n}{N}} \beta^{i-1}_n - \beta^i_n, & i =1, \ldots,
      N-1.
    \end{cases}
  \end{equation}
\end{dfn}
The root type generators are used to build the affine screening
currents.
\begin{dfn}
  We call the following vertex operators
  \begin{equation}
    \label{eq:11}
    S_i(z) = :\exp \left( -\sum_{n \neq 0} \frac{1}{n} \alpha^i_n
      z^{-n} \right):, \qquad i = 0,\ldots, N-1,
  \end{equation}
  affine screening currents.
\end{dfn} 

Combining the affine vertex operators~\eqref{eq:10} with affine
screening currents~\eqref{eq:11} one defines the screened affine
vertex operators
\begin{dfn}
  Let $x_i \in \mathbb{C}^*$, $i\in \mathbb{Z}$ be a
  set of parameters such that $x_{i+N} = x_i$. Thus, $x_1, \ldots,
  x_N$ define all $x_i$. Screened affine vertex operators are given by
  \begin{multline}
    \label{eq:12}
    \Psi^i(z|\vec{x},p) =\\
    = \sum_{\lambda} \left( \frac{\left( \frac{q}{t} ; q
        \right)_{\infty}}{(q;q)_{\infty}} \right)^{l(\lambda)}
    \prod_{k \geq 1} \left( p^{\frac{1}{N}}
      \frac{x_{N-i+k}}{x_{N-i+k-1}} \right)^{\lambda_k}
    \phi_{(i-l(\lambda))\mod N} \left( \kappa^{\frac{l(\lambda)+1}{N}}z
    \right) \prod_{j=1}^{
      \begin{smallmatrix}
        \leftarrow\\
        l(\lambda)
      \end{smallmatrix}
    } S_{(i-j + 1) \mod N} \left( \kappa^{\frac{j}{N}} q^{\lambda_j} z \right)
  \end{multline}
  where the sum is over all partitions (Young diagrams) $\lambda$,
  $l(\lambda)$ is the number of parts of the partition $\lambda$ (the
  length of the Young diagram) and the ordered product is defined as
  $\prod\limits_{j=1}^{
      \begin{smallmatrix}
        \leftarrow\\
        l
      \end{smallmatrix}
 } A_j = A_l \cdots A_1 $.
\end{dfn}
Finally, one can define the Shiraishi functions as the vacuum matrix
element of the product of screened affine vertex operators.
\begin{dfn}[\cite{Shir-func}]
  Define the shift operators $\mathbf{r}_i$ $i=1,\ldots,N$ acting on
  the $x_i$ variables as
  \begin{equation}
    \label{eq:15}
    \mathbf{r}_i x_j =
    \begin{cases}
      x_j, & j\leq i,\\
      tx_j, & j > i.
    \end{cases}
  \end{equation}
  Let
  \begin{equation}
    \label{eq:13}
    \widehat{f}_N(\vec{x}, p^{\frac{1}{N}} | \vec{s},
    \kappa^{\frac{1}{N}}|q,t) = \prod_{i<j} \frac{\left(
        t \frac{\tilde{s}_j}{\tilde{s}_i} ; q,\kappa \right)_{\infty} \left(
        \frac{q}{t} \frac{\tilde{s}_j}{\tilde{s}_i} ; q,\kappa \right)_{\infty}}{\left(
        \frac{\tilde{s}_j}{\tilde{s}_i} ; q,\kappa \right)_{\infty} \left(
        q \frac{\tilde{s}_j}{\tilde{s}_i} ; q,\kappa \right)_{\infty}}
    \langle 0 | \Psi^0(s_1| \mathbf{r}_N \vec{x}, p) \Psi^1(s_2|
    \mathbf{r}_{N-1} \vec{x}, p) \cdots \Psi^{N-1}(s_N| \mathbf{r}_1 \vec{x}, p) | 0\rangle,
  \end{equation}
  where $\tilde{s}_i = \kappa^{\frac{i}{N}} s_i$ and
  \begin{equation}
    \label{eq:14}
    (x;q,\kappa)_{\infty} = \prod_{n,m \geq 0} (1- q^n \kappa^m x).
  \end{equation}
  We call $\widehat{f}_N(\vec{x}, p^{\frac{1}{N}} | \vec{s},
  \kappa^{\frac{1}{N}}|q,t)$ the Shiraishi functions.
\label{dfn-shir-funct}
\end{dfn}

There are explicit expressions for Shiraishi's functions as series in
$\frac{x_i}{x_{i-1}}$ but we will not need them here.

Shiraishi's construction is very explicit but not too illuminating: at
more than one step one is tempted to ask why this concrete choice of
the parameters or relations for the operators were made and not the
other. In the next section we offer a way to view the affine vertex
operators and screenings as compositions of some elementary building
blocks, namely the intertwiners of the quantum toroidal algebra.

\subsection{Affine vertex operators and screening currents from quantum
toroidal algebra}
\label{sec:affine-vert-oper}

We construct both the ``weight type''~\eqref{eq:8} and ``root
type''~\eqref{eq:9} Heisenberg generators from $N$ independent
Heisenberg generators $a^{(i)}_n$ acting the tensor product
$\bigotimes\limits_{i=1}^N \mathcal{F}_{q,t^{-1}}^{(1,0)}\left(u_i
\right)$ of $N$ horizontal Fock representations of the algebra
$\mathcal{A}$ for certain choice of $u_i$. See
Appendix~\ref{sec:horiz-fock-repr} for the definitions of
$\mathcal{F}_{q,t^{-1}}^{(1,0)}\left(u_i \right)$. We have
\begin{equation}
  \label{eq:16}
  [a^{(i)}_n, a^{(j)}_m] = n \delta_{i,j} \frac{1-q^{|n|}}{1-t^{|n|}} \delta_{n+m,0}.
\end{equation}
\begin{lem} Given commutation relations~\eqref{eq:16} the combinations
  \begin{multline}
    \label{eq:17}
    \beta^i_n = \frac{1-t^n}{1-q^{-n}} \kappa^{-n\frac{i}{N}} \Biggl[
    - \sum_{j=1}^N a^{(j)}_n \left( \frac{t}{q} \right)^{n
      \frac{j-1}{2}} \frac{1 - \left( \frac{t}{q}
      \right)^{\frac{n+|n|}{2}} }{1 - \kappa^n} +\\
    + \sum_{k=1}^i \left( \frac{t}{q} \right)^{n \frac{k-1}{2}}
    a^{(k)}_n \left( 1 - \left( \frac{t}{q} \right)^{\frac{n+|n|}{2}}
    \right) + \left( \frac{t}{q} \right)^{n \frac{i}{2}} a^{(i+1)}_n
    \Biggr],
  \end{multline}
  for $i = 0,\ldots, N-1$ satisfy the commutation
  relations~\eqref{eq:8}.  \label{lem-1}
\end{lem}
\begin{cor}
  The Fock space generated by the ``weight type'' generators is
  isomorphic to the tensor product of horizontal Fock representations
  of $\mathcal{A}$,
  \begin{equation}
    \label{eq:35}
    \mathfrak{F}_N \simeq \bigotimes\limits_{i=1}^N
    \mathcal{F}_{q,t^{-1}}^{(1,0)}\left(u_i \right),
\end{equation}
and $|0\rangle = \bigotimes\limits_{i=1}^N \left|\varnothing,
  u_i \right\rangle$ for generic $u_i \in \mathbb{C}*$.
\end{cor}
The values of the spectral parameters $u_i$ in Eq.~\eqref{eq:35}
cannot be deduced from the identification~\eqref{eq:17} of the
Heisenberg generators alone. In Proposition~\ref{prop-aff-screened}
and Theorem~\ref{thm-shir-main} we specify $u_i$ in such a way that
the zero modes of a certain system of intertwiners of $\mathcal{A}$
reproduce the prefactors in the affine screened vertex
operators~\eqref{eq:12}.

The expression for the ``root type'' generators $\alpha^i_n$ in terms
of $a^{(i)}_n$ follows from the definition~\eqref{eq:9}:
\begin{equation}
  \label{eq:18}
  \alpha^i_n =
  \begin{cases}
    \frac{1 - t^n}{1 - q^{-n}} \left( \left( \frac{t}{q}
      \right)^{\frac{|n|+nN}{2}} \kappa^{-n} a^{(N)}_n - a^{(1)}_n
    \right), & i=0,\\
    \frac{1 - t^n}{1 - q^{-n}} \kappa^{-n \frac{i}{N}} \left(
      \frac{t}{q} \right)^{n \frac{i}{2}} \left( \left( \frac{t}{q}
      \right)^{\frac{|n|}{2}} a^{(i)}_n - a^{(i+1)}_n \right), & i=1,
    \ldots, N-1.
  \end{cases}
\end{equation}
Let us introduce the notation for a particular matrix element of the
$R$-matrix $\check{\mathcal{R}}|_{\mathcal{V}_q \otimes
  \mathcal{F}_{q,t^{-1}}^{(1,0)}(u)}$ corresponding to $n=0$ in
Eq.~\eqref{eq:31}.
\begin{dfn} Let $a_n$ satisfy~\eqref{eq:63}. We call
  \begin{multline}
    \label{eq:27}
    \bar{\mathbf{x}}_{q,t^{-1}}^q(w) \stackrel{\mathrm{def}}{=}
    (\ldots) \otimes \left\langle \left.  \sqrt{q/t}\, w
        \left|\check{\mathcal{R}}|_{\mathcal{V}_q \otimes
            \mathcal{F}_{q,t^{-1}}^{(1,0)}(u)} \right. \right| w
    \right\rangle \otimes
    (\ldots)=\\
    =\quad\includegraphics[valign=c]{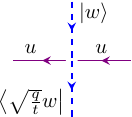}\quad
    = \exp \left[ \sum_{n \geq 1} \frac{w^{-n}}{n} \frac{1 - t^n}{1 -
        q^{-n}} \left( 1 - \left( \frac{t}{q} \right)^n \right) a_n
    \right]
  \end{multline}
  the crossing operator. 
\end{dfn}

To distinguish between the $R$-matrices, crossing operators and
intertwiners acting on each of $N$ horizontal Fock spaces
$\mathcal{F}_{q,t^{-1}}^{(1,0)}\left(u_i \right)$ we give them an
extra label $i = 1,\ldots, N$:
\begin{align}
  \label{eq:22}
    \check{\mathcal{R}}^{(i)} &\stackrel{\mathrm{def}}{=} \underbrace{1
  \otimes \ldots \otimes 1}_{i-1} \otimes
\check{\mathcal{R}}|_{\mathcal{V}_q \otimes
            \mathcal{F}_{q,t^{-1}}^{(1,0)}(u_i)} \otimes 1 \otimes
  \ldots \otimes 1,\\
  \bar{\mathbf{x}}_{q,t^{-1}}^{q,(i)}(w) &\stackrel{\mathrm{def}}{=} \underbrace{1
  \otimes \ldots \otimes 1}_{i-1} \otimes
  \bar{\mathbf{x}}_{q,t^{-1}}^q(w) \otimes 1 \otimes
  \ldots \otimes 1,\\
  \Phi^{q,(i)}_{q,t^{-1}}(w) &\stackrel{\mathrm{def}}{=} \underbrace{1
  \otimes \ldots \otimes 1}_{i-1} \otimes \Phi^{q}_{q,t^{-1}}(w) \otimes 1
\otimes \ldots \otimes 1,\\
  \Phi^{*q,(i)}_{q,t^{-1}}(w) &\stackrel{\mathrm{def}}{=} \underbrace{1
  \otimes \ldots \otimes 1}_{i-1} \otimes \Phi^{*q}_{q,t^{-1}}(w) \otimes 1
  \otimes \ldots \otimes 1,
\end{align}
e.g.~$\Phi^{(i)}$ acts only on the $i$-th horizontal Fock space and
therefore is written in terms of $a^{(i)}_n$ generators.

From Lemma~\ref{lem-1} and the
definitions~\eqref{eq:21},~\eqref{eq:27} we notice that the affine
vertex operators $\phi^i(w)$ for $i=0,\ldots,N-1$ can be expressed as
an infinite products of crossing operators
$\bar{\mathbf{x}}_{q,t^{-1}}^q$ and a single intertwiner
$\Phi^q_{q,t^{-1}}$.
\begin{lem}
  We have
\begin{multline}
  \label{eq:20}
  N(q,t,\kappa) \phi^i_0 \phi_i(w) =  \left\{ \prod_{m \geq 1} \prod_{j=1}^N
    \bar{\mathbf{x}}_{q,t^{-1}}^{q,(j)} \left( \kappa^{m+ \frac{i}{N}}
      \left( \frac{q}{t} \right)^{\frac{j-1}{2}}  w \right)
  \right\}\times\\
  \times \prod_{k=1}^i  \bar{\mathbf{x}}_{q,t^{-1}}^{q,(k)} \left( \kappa^{\frac{i}{N}}
    \left( \frac{q}{t} \right)^{\frac{k-1}{2}}  w \right) \Phi^{q,(i+1)}_{q,t^{-1}} \left(
    \kappa^{\frac{i}{N}} \left( \frac{q}{t} \right)^{\frac{i}{2}} w \right), 
\end{multline}
where $\phi^i_0 = t^{Q_i} \left( \kappa^{\frac{i}{N}}
  (q/t)^{\frac{i}{2}} w \right)^{\frac{P_i}{\ln q}}$ is the zero mode
part and
\begin{equation}
  \label{eq:38}
  N(q,t,\kappa) = \frac{\left(\kappa^{-1} ; q,\kappa^{-1}
    \right)_{\infty} \left(\kappa^{-1} q^{-1} ; q,\kappa^{-1}
    \right)_{\infty}}{\left(\kappa^{-1} \frac{t}{q} ; q,\kappa^{-1}
    \right)_{\infty}\left(\kappa^{-1} t^{-1} ; q,\kappa^{-1} \right)_{\infty}}.
\end{equation}
\label{lem-phi}
\end{lem}
Note that the ordering in the products in Eq.~\eqref{eq:20} is not
important since the crossing operators contain only positive
Heisenberg generators. Recalling the setup of~\cite{Zenkevich:2023cza}
we can interpret the infinite product of crossing operators in
Eq.~\eqref{eq:20} graphically as a spiralling ``tail'' attached to the
intertwiner $\Phi$ as follows:
\begin{equation}
  \label{eq:25}
  \phi^i_0 \phi_i \left( \kappa^{-\frac{i}{N}} \left( \frac{q}{t}
    \right)^{-\frac{i}{2}} w \right) =   \quad  \includegraphics[valign=c]{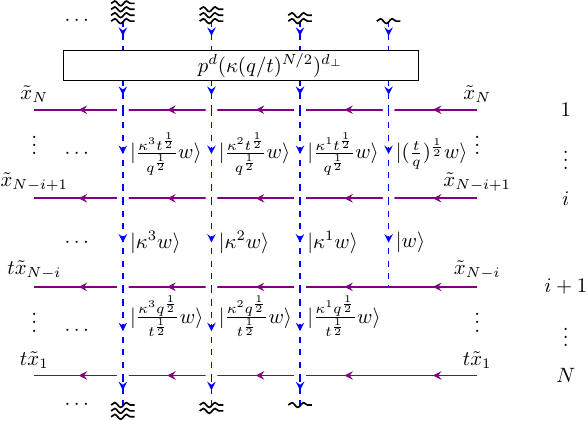}
\end{equation}
where $\tilde{x}_i = p^{\frac{i}{N}} x_i$ and the grading operators
$d$ and $d_{\perp}$ are defined in
Appendix~\ref{sec:quant-toro-algebra}. Note that in the r.h.s.\ of
Eq.~\eqref{eq:25} we assume that there is \emph{no summation} over the
intermediate states in the vector representations $\mathcal{V}_q$
(vertical dashed lines). To emphasize this we indicate a ket
\emph{state} over each intermediate dashed line. We will see in a
moment that the \emph{screened} affine vertex operator will correspond
precisely to adding back the sum over all the rest of the intermediate
states in the vector representations.

The relation~\eqref{eq:18} between the ``root type'' generators and
$a_n^{(i)}$ implies a simple construction of affine screening
currents~\eqref{eq:11} in terms of $\Phi$ and $\Phi^{*}$
intertwiners~\eqref{eq:21},~\eqref{eq:23}.
\begin{lem}
  We have
  \begin{equation}
    \label{eq:24}
    S^i_0 S_i(w) =
    \begin{cases}
      \Phi^{q,(1)}_{q,t^{-1}}(  w )
      \Phi^{*q,(N)}_{q,t^{-1}} \left( \kappa \left(
          \frac{q}{t} \right)^{\frac{N}{2}} w \right), & i =0,\\
      \Phi^{q,(i+1)}_{q,t^{-1}}\left( \kappa^{\frac{i}{N}} \left(
          \frac{q}{t} \right)^{\frac{i}{2}} w \right)
      \Phi^{*q,(i)}_{q,t^{-1}} \left( \kappa^{\frac{i}{N}} \left(
          \frac{q}{t} \right)^{\frac{i}{2}} w \right), & i=1,\ldots, N-1.
    \end{cases}
  \end{equation}
  where
  \begin{equation}
    \label{eq:26}
    S^i_0 =
    \begin{cases}
      t^{Q_1-Q_N}  w^{\frac{P_1 - P_N + \ln p + \ln t}{\ln
      q}} \left( \kappa (q/t)^{\frac{N}{2}} \right)^{\frac{\ln t -
      P_N}{\ln q}}, &i=0,\\
      t^{Q_{i+1}-Q_i} \left( \kappa^{\frac{i}{N}} (q/t)^{\frac{i}{2}}w \right)^{\frac{P_{i+1} - P_i + \ln t}{\ln
      q}},& i = 1,\ldots,N-1.
    \end{cases}
  \end{equation}
  is the zero mode part.
  \label{lem-S}
\end{lem}
We note that again no normal ordering is needed in the r.h.s.\ of
Eq.~\eqref{eq:24} since $\Phi$ and $\Phi^{*}$ act on different Fock
spaces. Eq.~\eqref{eq:24} can be depicted as the following diagram of intertwiners:
\begin{equation}
  \label{eq:29}
    S^0_0 S_0 (w) =
    \quad  \includegraphics[valign=c]{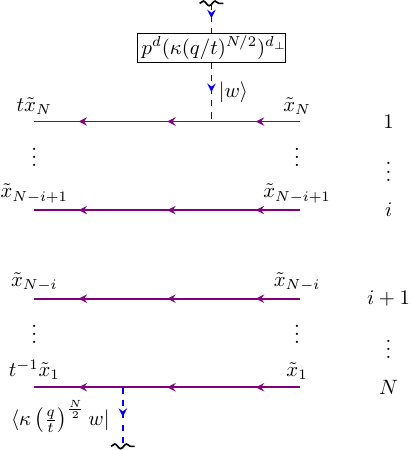},
  \end{equation}
  and
  \begin{equation}
  S^i_0 S_i \left(\kappa^{-\frac{i}{N}} \left( \frac{q}{t}
      \right)^{-\frac{i}{2}} w \right) =
    \quad  \includegraphics[valign=c]{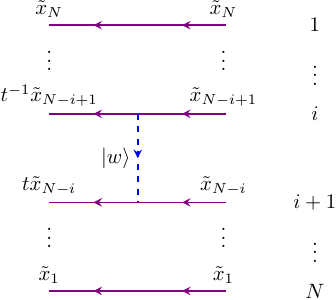},
    \qquad \text{for } i=1,\ldots,N-1.\label{eq:30}
\end{equation}

\begin{lem}[\cite{Shir-func}]
  We have
  \begin{equation}
    \label{eq:28}
    \phi_i(z) = :\phi_{(i-1) \mod N}\left(\kappa^{\frac{1}{N}}z \right) S_i(z):,
    \qquad i=0,\ldots, N-1.
  \end{equation}
\end{lem}
Eq.~\eqref{eq:28} follows from the definition of the ``root type''
generators~\eqref{eq:9}. In the intertwiner language it follows from
the identity~\eqref{eq:27}. Graphically Eq.~\eqref{eq:28} can be seen
as adding one segment (the screening current $S_i$) to the infinite
spiralling ``tail'' attached to the intertwiner
$\Phi^{q,(i-1)}_{q,t^{-1}}$.

Next let us consider the \emph{screened} affine vertex operators
$\Psi^i(z|\vec{x},p)$ as defined in Eq.~\eqref{eq:12}. They are given
by the sums of products of unscreened affine vertex operators and
affine screening currents. We are able to also express them through
intertwiners and $R$-matrices of the algebra $\mathcal{A}$.
\begin{prop}
  Using the identification~\eqref{eq:35} affine screened vertex
  operator $\Psi^i(z|\vec{x},p)$ can be viewed as an operator
  \begin{equation}
    \label{eq:42}
    \Psi^i(w|\mathbf{r}_{N-i}\vec{x},p): \bigotimes\limits_{j=1}^N
  \mathcal{F}_{q,t^{-1}}^{(1,0)}\left( t^{\delta_{j>i+1}}
    \tilde{x}_{N-j+1} \right) \to \bigotimes\limits_{j=1}^N
  \mathcal{F}_{q,t^{-1}}^{(1,0)}\left(t^{\delta_{j \geq i+1}}
    \tilde{x}_{N-j+1} \right),
\end{equation}
 where $\tilde{x}_i =
  x_i p^{\frac{i}{N}}$
  It is given by the following product of $R$-matrices and an
  intertwining operator:
  \begin{multline}
    \label{eq:37}
    N(q,t,\kappa) \Psi^i(w|\mathbf{r}_{N-i}\vec{x},p) =\\
    = \sum_{\lambda}   \prod_{m \geq 1}^{\leftarrow} \left\{ \prod_{j=1}^N
\left\langle  \kappa^{m+ \frac{i+1}{N}}
        \left( \frac{q}{t} \right)^{\frac{j}{2}} q^{\lambda_{i+1 + N
            m - j  }} w
    \right|\check{\mathcal{R}}^{(j)}\left| \kappa^{m+ \frac{i+1}{N}}
        \left( \frac{q}{t} \right)^{\frac{j-1}{2}} q^{\lambda_{i+2 + N
            m - j  }} w \right\rangle
    \right\} \times\\
    \times \prod_{k=1}^i   \left\langle  \kappa^{\frac{i+1}{N}}
      \left( \frac{q}{t} \right)^{\frac{k}{2}} q^{\lambda_{i-k+1}} w
    \right|\check{\mathcal{R}}^{(k)}\left| \kappa^{\frac{i+1}{N}}
      \left( \frac{q}{t} \right)^{\frac{k-1}{2}} q^{\lambda_{i-k+2}} w \right\rangle 
    \Phi^{q,(i+1)}_{q,t} \left(
      \kappa^{\frac{i+1}{N}} \left( \frac{q}{t} \right)^{\frac{i}{2}}
      q^{\lambda_1} w \right)
  \end{multline}
  for $i = 0,\ldots, N-1$.
  \label{prop-aff-screened}
\end{prop}
\begin{proof}
  The idea of the proof is to reorganize affine screening operators
  $S_i(z)$ and the affine vertex operator $\phi_i(z)$ in the
  definition~\eqref{eq:12} using the identifications~\eqref{eq:20}
  and~\eqref{eq:24} in terms of the crossing operators $\mathbf{x}$
  and intertwiners $\Phi$ and $\Phi^{*}$. One then uses
  Theorem~\ref{thm-r-q} to collect pairs of $\Phi$ and $\Phi^{*}$
  acting on the same Fock space into the $R$-matrices.

  The prefactor $\left( \frac{\left( \frac{q}{t} ; q
      \right)_{\infty}}{(q;q)_{\infty}} \right)^{l(\lambda)}$
  in~\eqref{eq:12} originates from applying Theorem~\ref{thm-r-q} to
  every pair of adjacent affine screening currents and also the last
  screening and the intertwiner. The prefactor $\left( p^{\frac{1}{N}}
    \frac{x_{N-i+k}}{x_{N-i+k-1}} \right)^{\lambda_k} = \left(
    \frac{\tilde{x}_{N-i+k}}{\tilde{x}_{N-i+k-1}} \right)^{\lambda_k}$
  in~\eqref{eq:12} appears from the zero mode parts $\phi_i^0$ and
  $S_i^0$ featuring in the Lemmas~\ref{lem-phi} and~\ref{lem-S}
  respectively.
\end{proof}

\begin{rem}
  The sum over Young diagrams in Eq.~\eqref{eq:37} can be extended to
  the sum over all finite length sequences of integers since
  Theorem~\ref{thm-r-q} automatically enforces the Young diagram
  condition $\lambda_i \geq \lambda_{i+1}$ for all $i \in
  \mathbb{N}$. Thus the sum in Eq.~\eqref{eq:37} is the sum over the
  complete basis of states in the intermediate vector representations
  $\mathcal{V}_q$ between the adjacent $R$-matrices in the product and
  between the last $R$-matrix in the product and the intertwiner
  $\Phi$.
\end{rem}

\begin{rem}
  Notice that $\Psi^i(z|\vec{x},p)$ shifts the spectral parameter
  $\tilde{x}_{i+1}$ of the $(i+1)$-th Fock space in the tensor product
  by $t$. One can check that it is consistent with the action of the
  shift operators $\mathbf{r}_i$ on $x_j$ in the
  Definition~\ref{dfn-shir-funct} of the Shiraishi functions.
\end{rem}

This means that $\Psi^i(w|\vec{x},p)$ is a matrix element of a
composition of an infinite number of intertwining operators of the
algebra $\mathcal{A}$. This will be crucial for understanding the
properties of Shiraishi functions in what follows.

It is perhaps easier to understand the structure of the
formula~\eqref{eq:37} by drawing the corresponding picture of
intertwiners:
\begin{multline}
  N(q,t,\kappa) \Psi^i\left( \left.  \kappa^{-\frac{i+1}{N}} \left( \frac{q}{t} \right)^{-\frac{i}{2}}
      w \right|\mathbf{r}_{N-i}\vec{x},p \right) =\\
  = \sum_{\lambda}
  \quad  \includegraphics[valign=c]{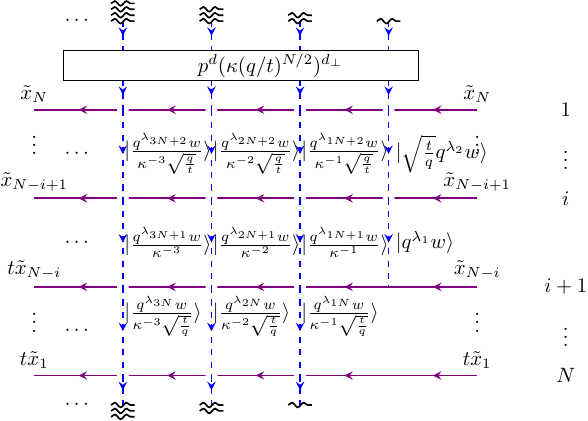}\quad
  = \notag
\end{multline}
\begin{equation}
=  \quad  \includegraphics[valign=c]{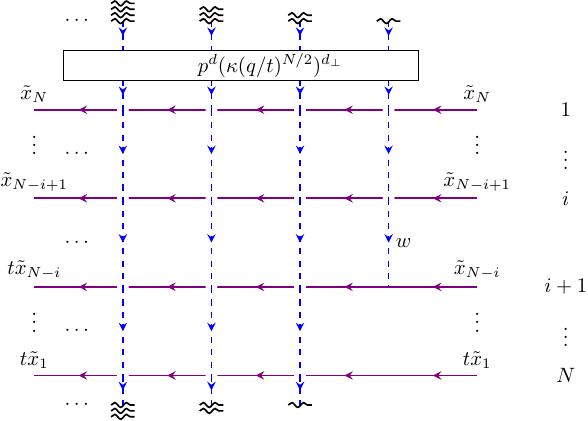}\label{eq:34}
\end{equation}
In the first picture in Eq.~\eqref{eq:34} we write out explicitly the
sum over intermediate states in the vector representations (hence the
ket states on the dashed lines and the sum over $\lambda$). In the
second picture the sums are implicit and we the intermediate dashed
lines are just compositions of operators along $\mathcal{V}_{q_1}$
representation (see similar discussion after Eq.~\eqref{eq:25}).

Let us introduce a graphical abbreviation for the
picture~\eqref{eq:34} (notice the shift of $w$ on both sides of the
equality):
\begin{equation}
  \label{eq:40}
  N(q,t,\kappa) \Psi^i(w|\mathbf{r}_{N-i}\vec{x},p) = \quad  \includegraphics[valign=c]{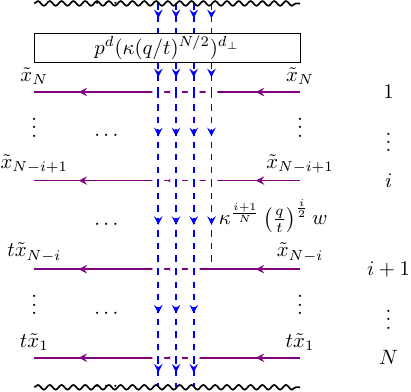} 
\end{equation}

Using the expression~\eqref{eq:37} for affine screened vertex
operators we obtain the main theorem of this section expressing the
Shiraishi functions as intertwiners of the quantum toroidal algebra.
\begin{thm}
  Shiraishi function is given by the following vacuum-vacuum matrix
  element of intertwiners and $R$-matrices of the algebra
  $\mathcal{A}$:
  \begin{multline}
    \label{eq:41}
      \widehat{G}_{q,t^{-1}}(\vec{s})\, \widehat{f}_N(\vec{x}, p^{\frac{1}{N}} | \vec{s},
      \kappa^{\frac{1}{N}}|q,t) =\\
      \\
    = \quad  \includegraphics[valign=c]{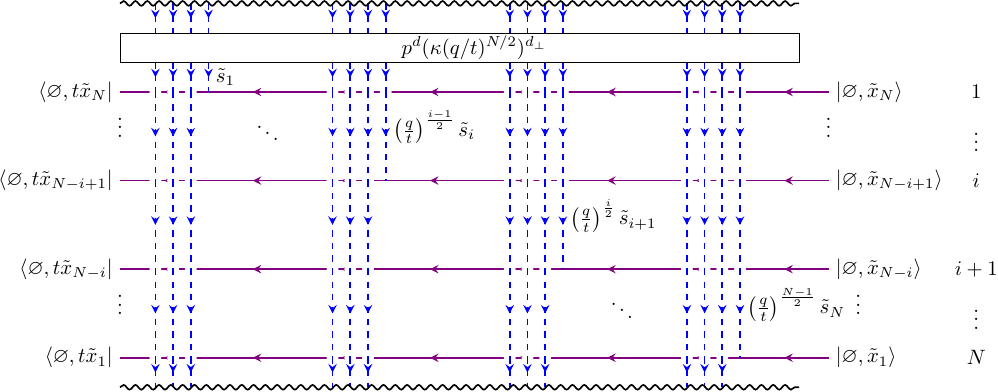}
  \end{multline}
  where $\tilde{s}_i = \kappa^{\frac{i}{N}} s_i$,
  \begin{equation}
    \label{eq:120}
    \widehat{G}_{q,t^{-1}}(\vec{s}) =   (N(q,t,\kappa))^N  \prod_{i<j} \frac{\left(
        \frac{\tilde{s}_j}{\tilde{s}_i} ; q,\kappa \right)_{\infty} \left(
        q \frac{\tilde{s}_j}{\tilde{s}_i} ; q,\kappa
      \right)_{\infty}}{\left( t
        \frac{\tilde{s}_j}{\tilde{s}_i} ; q,\kappa \right)_{\infty} \left(
        \frac{q}{t}  \frac{\tilde{s}_j}{\tilde{s}_i} ; q,\kappa
      \right)_{\infty}},
  \end{equation}
  and we have used the
  graphical abbreviations~\eqref{eq:40}.
\label{thm-shir-main}
\end{thm}
\begin{proof}
  The product of screened affine vertex operators in
  Definition~\ref{dfn-shir-funct} is manifestly reproduced by the
  composition of blocks of dashed lines in the picture. Notice that
  the $t$-shifts in the spectral parameters of the horizontal lines
  (see Eq.~\eqref{eq:42}) are compatible with taking the compositions.
\end{proof}
\begin{rem}
  The prefactor $\widehat{G}_{q,t^{-1}}(\vec{s})$ in the definition of
  the Shiraishi functions~\eqref{eq:41} is equal to the term with
  trivial partitions $\lambda^{(i)}$ in $\Psi^i$. Thus, Shiraishi
  functions $\widehat{f}_N(\vec{x}, p^{\frac{1}{N}} | \vec{s},
  \kappa^{\frac{1}{N}}|q,t)$ can be written as a sum over $N$
  partitions $\lambda^{(i)}$ with the first term being identity.
\end{rem}

It will be useful in what follows to redraw the system of
intertwiners~\eqref{eq:41} in an equivalent way. We notice that the
diagram~\eqref{eq:41} can be viewed as drawn on a cylinder, so that
the $i$-th group of dashed lines (a ``coil'') starts on the $i$-th
horizontal line and winds around it an infinite number of
times. Instead of taking a composition of ``coils'' we \emph{first}
attach the $i$-th intertwiners $\Phi^{q_1,(i)}_{q_1,q_2}$ to the
$i$-th horizontal line and then wind the resulting $N$ dashed
``tails'' around the cylinder \emph{together} as a group:
\begin{equation}
  \label{eq:114}
  \includegraphics[valign=c]{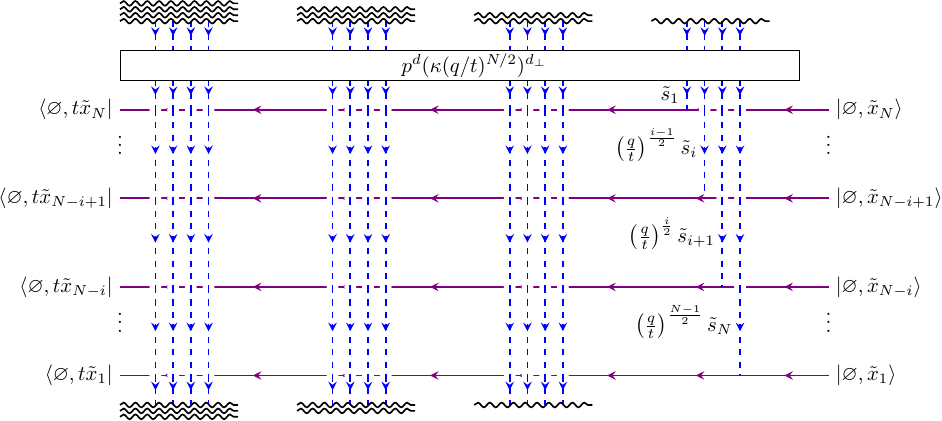}
\end{equation}
\label{conj-shir}
\begin{conj}
  The diagrams~\eqref{eq:114} and~\eqref{eq:41} evaluate to the same
  result.
  \label{conj-redraw}
\end{conj}
We hope that one can use the intertwining property of the $\Phi$ and
the Yang-Baxter equation satisfied by the $R$-matrices to prove
Conjecture~\ref{conj-redraw}. However, the proof would require a
careful analysis of an infinite number of commutations between the
intertwiners, so we leave it for the future.

\section{Affine $qq$-characters and noncommutative Jacobi identity
  from spiralling branes}
\label{sec:nonc-jacobi-ident}

In this section we use the spiralling brane construction incorporating
Miura $R$-matrices~\eqref{eq:5} to prove the general version of the
noncommutative Jacobi identity for $qq$-characters of affine
$\widehat{A}_0$ type recently proposed in~\cite{Grekov:2024ayn} (see
also~\cite{Grekov:2024fyh,Grekov:2023psy,Grekov:2023fek}). In~\cite{Grekov:2024ayn}
the identity was proven only in the Nekrasov-Shatashvili limit which
corresponds to the $q$-character limit of the $qq$-characters. Our
approach is general and provides the connection with representation
theory of the quantum toroidal algebra $\mathcal{A}$.

\subsection{Fundamental $\widehat{A}_0$ $qq$-character}
\label{sec:fundamntal-a_0-qq}
\begin{dfn}
  Let $\tilde{q}_3, \tilde{q}_4 \in \mathbb{C}^{*}$ be
  parameters satisfying
  \begin{equation}
    \label{eq:87}
    q_1 q_2 \tilde{q}_3 \tilde{q}_4 = 1,
  \end{equation}
  but otherwise generic. Fundamental $qq$-character of $\widehat{A}_0$
  type
  \begin{equation}
    \label{eq:80}
    \mathsf{X}(z|q_1,q_2, \tilde{q}_3, \tilde{q}_4, p) :
    \mathcal{F}_{q_1,q_2}^{(1,0)}(v) \to \mathcal{F}_{q_1,q_2}^{(1,0)}(v)
  \end{equation}
  is given by  
  \begin{equation}
    \label{eq:78}
    \mathsf{X}(z|q_1,q_2, \tilde{q}_3, \tilde{q}_4, p) = \sum_{\lambda} p^{|\lambda|}
    z_{\lambda}^{\mathrm{adj}} (\tilde{q}_3, \tilde{q}_4,q_1,
    q_2)\, : \prod_{(a,b)\in
      \mathrm{CC}(\lambda)}
    \mathsf{Y}(\tilde{q}_3^a \tilde{q}_4^b z) \prod_{(c,d)\in
      \mathrm{CV}(\lambda)}
    \left(\mathsf{Y}(\tilde{q}_3^c \tilde{q}_4^d z)\right)^{-1}:,
  \end{equation}
  where $\lambda$ is a Young diagram, $|\lambda|$ its total number of
  boxes, $\mathrm{CC}(\lambda)$ (resp.~$\mathrm{CV}(\lambda)$) is the
  set of concave (convex) corners of $\lambda$ (see Fig.~\ref{fig:1}),
  \begin{equation}
    \label{eq:84}
    z_{\lambda}^{\mathrm{adj}} (\tilde{q}_3, \tilde{q}_4,q_1,
    q_2) = \prod_{(a,b)\in \lambda} \psi_{q_1, q_2} \left(
      \tilde{q}_3^{\mathrm{Leg}_{\lambda}(a,b)+1} \tilde{q}_4 ^{- \mathrm{Arm}_{\lambda}(a,b)}\right),
  \end{equation}
  where
  \begin{equation}
    \label{eq:101}
    \mathrm{Arm}_{\lambda}(a,b) = \lambda_a - b, \qquad
    \mathrm{Leg}_{\lambda}(a,b) = a - \lambda^{\mathrm{T}}_b,
  \end{equation}
  $\psi_{q_1, q_2}(z)$ is given by Eq.~\eqref{eq:51} and
  \begin{equation}
    \label{eq:81}
    \mathsf{Y}(z) =\, :\exp \left[ \sum_{n \neq 0} \frac{\mathsf{y}_n}{n}
      z^{-n} \right]:,
  \end{equation}
  where
\begin{equation}
  \label{eq:86}
  \mathsf{y}_n =
  \begin{cases}
    \frac{1-q_2^{-n}}{1-\tilde{q}_3^{-n}} a_n^{(q_1,q_2)}, &n>0,\\
    \frac{1-q_2^{-n}}{1-\tilde{q}_4^n} a_n^{(q_1,q_2)}, &n<0.
  \end{cases}
\end{equation}
and $a_n^{(q_1,q_2)}$ satisfy the commutation relations~\eqref{eq:44}.
\label{dfn-qq}
\end{dfn}

\begin{rem}
  It is important not to confuse $q_3 = (q_1 q_2)^{-1}$ and
  $\tilde{q}_3$ which is a parameter independent of $q_1$, $q_2$. Note
  that $\tilde{q}_3 \tilde{q}_4 = q_3$.
\end{rem}
\begin{rem}
  The generators $\mathsf{y}_n$ satisfy the commutation relations
  \begin{equation}
    \label{eq:79}
    [\mathsf{y}_n, \mathsf{y}_m] = n \frac{(1 - q_1^n)(1-q_2^n)}{(1-\tilde{q}_3^{-n})(1-\tilde{q}_4^{-n})} \delta_{n+m,0}.
  \end{equation}
\end{rem}
\begin{rem}
  In physical terms $z_{\lambda \lambda}^{\mathrm{adj}} (\tilde{q}_3,
  \tilde{q}_4,q_1, q_2)$ is the Nekrasov factor for a $5d$ $U(1)$
  gauge theory in an $\Omega$-background $\mathbb{C}_{\tilde{q}_3}
  \times \mathbb{C}_{\tilde{q}_4} \times S^1$ with adjoint matter
  hypermultiplet of exponentiated mass $ \sqrt{\frac{q_1}{q_2}}$.
\end{rem}

We introduce the spiral transfer matrix which will play the main part
in our proof of the noncommutative Jacobi identity.
\begin{dfn}
  We call the spiral transfer matrix
  \begin{equation}
    \label{eq:68}
    \widehat{R}^{q_3, M}_{q_1,q_2}(x|v|p,\mu): \mathcal{V}^{*}_{q_3} \otimes
    \mathcal{F}_{q_1,q_2}^{(1,0)}(v) \to 
    \mathcal{F}_{q_1,q_2}^{(1,0)}(v) \otimes \mathcal{V}^{*}_{q_3}
  \end{equation}
  the following twisted product of Miura $R$-matrices
  \begin{multline}
  \label{eq:66}
 \widehat{R}^{q_3, M}_{q_1,q_2}(x|v|p,\mu) = p^{\frac{M^2}{2}}
 \left(-\frac{v}{q_3}\right)^M P \left(p^{\left(M - \frac{1}{2}\right) d}
 \left(\frac{\mu}{\sqrt{q_3}} \right)^{-M d_{\perp}} \otimes
    1 \right) \left(\mathcal{R}|_{\mathcal{V}^{*}_{q_3} \otimes
      \mathcal{F}_{q_1,q_2}^{(1,0)}(v)}(p^{-d} \mu^{d_{\perp}} \otimes
    1) \right)^{2 M} \times\\
    \times  \left( p^{\left(M + \frac{1}{2}\right) d} (\sqrt{q_3} \mu)^{-Md_{\perp}} \otimes
      1 \right) 
\end{multline}
where $d$ and $d_{\perp}$ are the grading operators of the algebra
$\mathcal{A}$ (see Eqs.~\eqref{eq:71}--\eqref{eq:70}) and $P$ is the operator permuting the two factors in the
tensor product.
\end{dfn}
\begin{rem}
  In Eq.~\eqref{eq:66} we use the $R$-matrix $\mathcal{R}$ which
  differs from $\check{\mathcal{R}} = P \mathcal{R}$ by the
  permutation operator $P$ acting on the two tensor factors (see
  Appendix~\ref{sec:universal-r-matrix-1}).
\end{rem}
\begin{rem}
  The grading operators before and after the product of $R$-matrices
  in Eq.~\eqref{eq:66} are a matter of convenience. They produce an
  overall shift in $v$ and $x$ parameters.
\end{rem}
\begin{rem}
  The spiral transfer matrix can be drawn using the diagram notation
  as follows:
  \begin{equation}
    \label{eq:75}
    \widehat{R}^{q_3, M}_{q_1,q_2}(x|v|p,\mu)    = p^{\frac{M^2}{2}} \left(-
      \frac{v}{q_3} \right)^M \times    \quad  \includegraphics[valign=c]{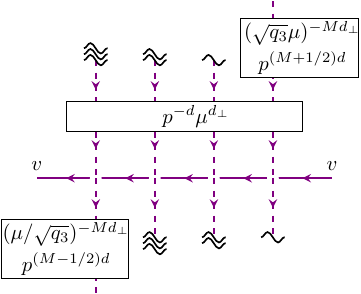}
  \end{equation}
\end{rem}

The noncommutative Jacobi identity follows from the relation between
the spiral transfer matrix and the fundamental $qq$-character of type
$\widehat{A}_0$.

\begin{prop}
  Let $|\mu|>|q_3|^{\pm 1/2}$. Then the limit
  \begin{equation}
    \label{eq:83}
    \widehat{R}^{q_3, \infty}_{q_1,q_2}(x|v|p,\mu) = \lim_{M \to \infty} \widehat{R}^{q_3, M}_{q_1,q_2}(x|v|p,\mu)
  \end{equation}
  makes sense as a formal power series in $p^{1/2}$.
  \label{prop-infty}
\end{prop}
  
\begin{thm}[Noncommutative Jacobi identity]
  \begin{equation}
    \label{eq:82}
    \widehat{R}^{q_3, \infty}_{q_1,q_2}(x|v|p,\mu)  = \sum_{k \in \mathbb{Z}} p^{\frac{k^2}{2}}
    (-\tilde{v})^k \mathsf{X}(\tilde{q}_4^k x|q_1,q_2, \tilde{q}_3, \tilde{q}_4, p) q_3^{k x \partial_x},
  \end{equation}
  where $\tilde{v} = \frac{v}{q_3}$,
  \begin{align}
    \label{eq:73}
    \tilde{q}_3 &= \sqrt{q_3} \mu,\\
    \tilde{q}_4 &= \sqrt{q_3} \mu^{-1},\label{eq:97}
  \end{align}
  and the r.h.s.\ of Eq.~\eqref{eq:82} is viewed as an
  element of $\mathrm{Hom} \left( \mathcal{V}^{*}_{q_3} \otimes
    \mathcal{F}_{q_1,q_2}^{(1,0)}(v).
    \mathcal{F}_{q_1,q_2}^{(1,0)}(v) \otimes \mathcal{V}^{*}_{q_3}
  \right)$.
  \label{thm-jacobi}
\end{thm}

\begin{proof}
  We prove Proposition~\ref{prop-infty} and Theorem~\ref{thm-jacobi}
  together. Plugging the explicit explicit expression for the Miura
  $R$-matrix from Theorem~\ref{thm-miura-r} into the definition of the
  spiral transfer matrix~\eqref{eq:68} we get
\begin{multline}
  \label{eq:67}
  \widehat{R}^{q_3, M}_{q_1,q_2}(x|v|p,\mu) =\\
  = p^{\frac{M^2}{2}}
 \left(- \tilde{v}\right)^M x^{\left( M  -\frac{1}{2} \right) \frac{\ln p}{\ln q_3}} \left(\frac{\mu}{\sqrt{q_3}}\right)^{-M x \partial_x}
  \left( R_{q_1,q_2}^{q_3}(x|v) x^{-\frac{\ln p}{\ln
        q_3}} \mu^{x \partial_x}
  \right)^{2M} x^{\left( M  +\frac{1}{2} \right) \frac{\ln p}{\ln q_3}}
  \left(\sqrt{q_3} \mu\right)^{-M x \partial_x} = \\
  = p^{\frac{M^2}{2}} \prod_{i=1}^{
    \begin{smallmatrix}
      \to\\
      2M
    \end{smallmatrix}
    } \left( \tilde{v}^{\frac{1}{2}} p^{-\frac{i-M}{2} +
    \frac{1}{4}}
    V^{+}_{q_1,q_2}(\mu^{i-M-1} q_3^{M/2}  x) q_3^{\frac{1}{2}
      x \partial_x} - \tilde{v}^{-\frac{1}{2}} p^{\frac{i-M}{2}-\frac{1}{4}}
    V^{-}_{q_1,q_2}(\mu^{i-M-1} q_3^{M/2})
    q_3^{-\frac{1}{2} x \partial_x} \right),
\end{multline}
where $V^{\pm}_{q_1,q_2} (x)$ are given by
Eq.~\eqref{eq:43}.
In this section from now on we will abbreviate $V^{\pm}_{q_1,q_2} (x)$
as $V_{\pm}(x)$.

Expanding the brackets in the product in Eq.~\eqref{eq:67} we get
$2^{2M}$ terms, each corresponds a sequence of signs $(s_1, \ldots,
s_{2M})$, $s_i \in \{\pm 1\}$. Indeed in each bracket we can pick
either the $V_{+}$ or $V_{-}$ vertex operator.  The sequence of signs
can be drawn as a Maya diagram (see Fig.~\ref{fig:1}). We have:
\begin{equation}
  \label{eq:85}
\widehat{R}^{q_3, M}_{q_1,q_2}(x|v|p,\mu) = \sum_{(s_i) \in \{\pm
  1\}^{2M}} 
\tilde{v}^{\frac{1}{2}\sum_{i=1}^{2M} s_i} p^{\frac{M^2}{2} + \sum_{i=1}^{2M} s_i
  \left( \frac{1}{4} + \frac{M-i}{2} \right)}  \prod_{i=1}^{
    \begin{smallmatrix}
      \to\\
      2M
    \end{smallmatrix}
} \left( s_i V_{s_i} (\mu^{i-M-1} q_3^{M/2}  x)  q_3^{\frac{s_i}{2} x \partial_x} \right) .
\end{equation}

\begin{figure}[h]
  \centering
  \includegraphics[valign=c]{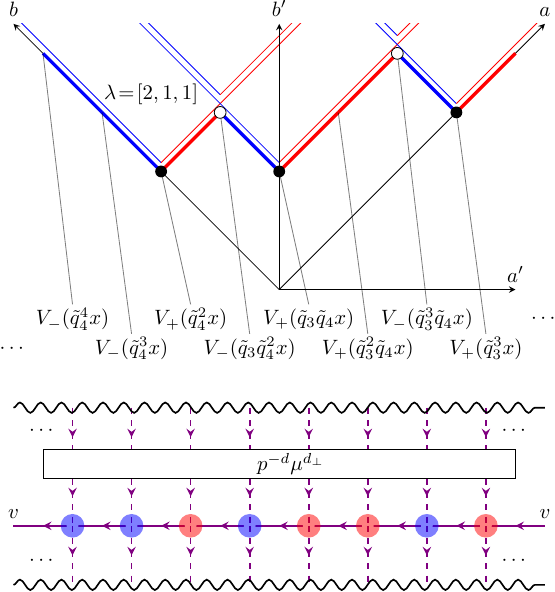}
  \caption{A term in the expansion of the spiral transfer matrix
    $\widehat{R}^{q_3, \infty}_{q_1,q_2}(x|v|p,\mu)$ as a Maya diagram
    and the corresponding Young diagram $\lambda = [2,1,1]$ (the level
    $k = 0$ in this example). Blue (resp.\ red) dots in the Maya
    diagram correspond to $V_{-}$ (resp.\ $V_{+}$) vertex operators in
    the product. The positions of the $V_{\pm}$ insertions are
    $\tilde{q}_3^x \tilde{q}_4^y$ with $(x,y)$ coordinates of a point
    on the boundary of $\lambda$. Black (resp.\ white) dots are
    concave (resp.\ convex) corners of $\lambda$. Thin angles
    suspended above black (resp.\ white) dots represent the operators
    $\mathsf{Y}(\tilde{q}_3^x \tilde{q}_4^y)$ (resp.\
    $\left(\mathsf{Y}(\tilde{q}_3^x \tilde{q}_4^y)\right)^{-1}$) which are
    infinite products of $V_{\pm}$.}
  \label{fig:1}
\end{figure}

We use the standard identification between Maya diagrams and Young
diagrams (see Fig.~\ref{fig:1}). A sequence of signs $(s_i)$
corresponds to a pair $(k,\lambda)$ consisting of a level $k = -M,
\ldots, M$ and a Young diagram $\lambda$ fitting inside a rectangular
$(k-M)\times (k+M)$ box. The level is given by the difference between
the total number of minus and plus signs:
\begin{equation}
  \label{eq:90}
  \frac{1}{2} \sum_{i=1}^{2M} s_i = -k.
\end{equation}
We also use the standard formula for the energy of the fermionic state
labelled by the Maya diagram $(s_i)$:
\begin{equation}
\sum_{i=1}^{2M} i s_i = M^2 - k^2 -k(2M+1) - 2|\lambda|.\label{eq:91}
\end{equation}
Plugging Eqs.~\eqref{eq:90}, \eqref{eq:91} into Eq.~\eqref{eq:85} we
get
\begin{equation}
  \label{eq:92}
  \widehat{R}^{q_3, M}_{q_1,q_2}(x|v|p,\mu) = \sum_{k = -M}^M
  \sum_{\lambda \in P_{(k-M)\times (k+M)}} (-\tilde{v})^{-k}
  p^{\frac{k^2}{2} + |\lambda|} \prod_{i=1}^{
    \begin{smallmatrix}
      \to\\
      2M
    \end{smallmatrix}
} \left(  V_{s_i} (\mu^{i-M-1} q_3^{M/2}  x)  q_3^{\frac{s_i}{2} x \partial_x} \right),
\end{equation}
where $P_{(k-M)\times (k+M)}$ is the set of Young diagrams fitting
inside the $(k-M)\times (k+M)$ box and we have kept the Maya diagram
variables $s_i$ in the vertex operator product. From Eq.~\eqref{eq:92}
we deduce that the number of terms with a given power $K$ of $p$ is
finite and independent of $M$ for $M$ sufficiently large. Hence we can
consider the limit $M \to \infty$ in each term and then sum over $k$
and $\lambda$.

Let us move the $q_3$-shift operators in the vertex operator product
to the right. This will introduce a shift $q_3^{\frac{1}{2}\sum_{j<i}
  s_i}$ in the argument of the $i$-th vertex operator. We notice that
the resulting argument
\begin{equation}
  \label{eq:94}
  q_3^{\frac{M}{2} + \frac{1}{2}\sum_{j<i} s_i} \mu^{i-M-1}x =
  \mu^{a'_i+k} q_3^{b'_i - k/2} x, 
\end{equation}
where $(a'_i, b'_i)$ are coordinates of the integer point on the
boundary $\partial \lambda$ of the Young diagram $\lambda$ (rotated by
$\frac{\pi}{4}$ and considered together with the coordinate axes $a$
and $b$ as in Fig.~\ref{fig:1}) corresponding to the sequence
$(s_i)$. One also associate the sign $s(a',b') = \frac{db'(a'+0)}{da'} =
s_i$ to every integer point of $\partial \lambda$. We find
\begin{equation}
  \label{eq:95}
  \prod_{i=1}^{
    \begin{smallmatrix}
      \to\\
      2M
    \end{smallmatrix}
  } \left(  V_{s_i} (\mu^{i-M-1} q_3^{M/2}  x)  q_3^{\frac{s_i}{2}
      x \partial_x} \right) = \prod_{(a',b') \in \partial \lambda}^{
    \to} \left(  V_{s(a',b')} (\mu^{a'+k} q_3^{b' - k/2}  x) \right)   q_3^{-k x \partial_x}
\end{equation}
where the ordering in the product is by the increasing value of $x'$.

Let us normal order the $V_{\pm}$ vertex operators in the product
in~\eqref{eq:95} using the commutation relations~\eqref{eq:50}. In
this way one gets a factor $\psi_{q_1,q_2}$ for every pair of $V_{+}$
and $V_{-}$ in which $V_{+}$ is to the left of $V_{-}$. Such pairs
correspond to pairs of integer points on $\partial \lambda$ such that
$s(A)=+1$, $s(B)=-1$ and $a'(A) < a'(B)$. In Fig.~\ref{fig:1} this
corresponds to $A$ living on a red segment of the boundary and $B$ on
a blue one. We notice that such pairs $(A,B)$ are in one-to-one
correspondence with points inside $\lambda$. Indeed, starting with
$(A,B)$ as described above one can draw a line parallel to the $b$
axis from $A$ and a line parallel to the $a$ axis from $B$. These
lines intersect at an integer point inside $\lambda$. Vice versa
starting with an integer point in $\lambda$, drawing lines parallel to
the coordinate axes $a$ and $b$ and intersecting them with $\partial
\lambda$ one obtains a pair of boundary points $(A,B)$ satisfying
$s(A)=+1$, $s(B)=-1$ and $a'(A) < a'(B)$.

According to the commutation relations~\eqref{eq:50} the factor
corresponding to a pair of boundary points $(A,B)$ is
\begin{equation}
  \label{eq:89}
  \psi_{q_1,q_2} \left(\mu^{a'(B) - a'(A)}
   q_3^{\frac{1}{2}(b'(B)-b'(A))} \right) = \psi_{q_1, q_2} \left(
      \tilde{q}_3^{\mathrm{Leg}_{\lambda}(i,j)+1} \tilde{q}_4 ^{- \mathrm{Arm}_{\lambda}(i,j)}\right).
\end{equation}
where $(i,j) \in \lambda$ is the point corresponding to the pair of
boundary points $(A,B)$ and $\tilde{q}_3$, $\tilde{q}_4$ are defined
by Eqs.~\eqref{eq:73}, \eqref{eq:97}.

This gives us the following normal ordered expression for the product
of vertex operators from Eq.~\eqref{eq:95}:
\begin{multline}
  \label{eq:93}
  \prod_{(a',b') \in \partial \lambda}^{ \to} \left( V_{s(a',b')}
    (\mu^{a'+k} q_3^{b' - k/2} x) \right)
  =\\
  =\prod_{(i,j)\in \lambda} \psi_{q_1, q_2} \left(
    \tilde{q}_3^{\mathrm{Leg}_{\lambda}(i,j)+1} \tilde{q}_4 ^{-
      \mathrm{Arm}_{\lambda}(i,j)}\right) \, :\prod_{(a',b') \in \partial
    \lambda} \left( V_{s(a',b')} (\mu^{a'+k} q_3^{b' - k/2} x)
  \right):\,=\\
  =z_{\lambda}^{\mathrm{adj}} (\tilde{q}_3, \tilde{q}_4,q_1,
    q_2) \, :\prod_{(a',b') \in \partial
    \lambda} \left( V_{s(a',b')} (\mu^{a'+k} q_3^{b' - k/2} x)
  \right):.
\end{multline}

Finally, we reorganize the normal ordered product in Eq.~\eqref{eq:93}
in terms of $\mathsf{Y}(z)$ introduced in Eq.~\eqref{eq:81}. To this
end let us first consider the contribution of $\lambda =
\varnothing$. In our convention the boundary $\partial \varnothing$ is
given by the union of $a$ and $b$ coordinate axes. We have
\begin{multline}
  \label{eq:96}
  :\prod_{(a',b') \in \partial \varnothing} \left( V_{s(a',b')}
    (\mu^{a'+k} q_3^{b' - k/2} x) \right): \, = \, :\prod_{a =
    1}^{M+k} V_{-} (\tilde{q}_4^{a-k}x) \prod_{b
    = 0}^{M-k-1}  V_{+} (\tilde{q}_3^b \tilde{q}_4^{-k} x):\, =\\
  = \exp \left[ - \sum_{n \geq 1}\frac{1-q_2^n}{n} \sum_{a=1}^{M+k}
    \left( \tilde{q}_4^{a-k} x \right)^n a_{-n}^{(q_1,q_2)} \right]
  \exp \left[ - \sum_{n \geq 1}\frac{1-q_2^{-n}}{n} \sum_{a=1}^{M-k-1}
    \left( \tilde{q}_3^b \tilde{q}_4^{-k} x \right)^{-n}
    a_n^{(q_1,q_2)} \right],
\end{multline}
where $a_n^{(q_1,q_2)}$ are Heisenberg generators
satisfying~\eqref{eq:44}. The assumption $|\mu| > |q_3|^{\pm 1/2}$ in
Proposition~\ref{prop-infty} means that $|\tilde{q}_3|>1$,
$|\tilde{q}_4|<1$ which implies that the limit $M \to \infty$ of the
product exists and is given by
\begin{multline}
  \label{eq:98}
  \lim_{M \to \infty} \exp
  \left[ - \sum_{n \geq 1}\frac{1-q_2^n}{n} \sum_{a=1}^{M+k} \left(
      \tilde{q}_4^{a-k} x \right)^n a_{-n}^{(q_1,q_2)} \right] \exp
  \left[ - \sum_{n \geq 1}\frac{1-q_2^{-n}}{n} \sum_{a=1}^{M-k-1} \left(
      \tilde{q}_3^b \tilde{q}_4^{-k} x \right)^{-n} a_n^{(q_1,q_2)}
  \right] =\\
  =\mathsf{Y}(\tilde{q}_4^{-k} x).
\end{multline}
Thus, $\mathsf{Y}(x)$ is given by an infinite product
\begin{equation}
  \label{eq:99}
  \mathsf{Y}(x) = \, :\prod_{a \geq
    1} V_{-} (\tilde{q}_4^ax) \prod_{b
    \geq 0}  V_{+} (\tilde{q}_3^b x):.
\end{equation}
Taking into account that the operators under the normal ordering
commute we can express the product of $V_{\pm}(\tilde{q}_3^a
\tilde{q}_4^b x) $ with $(x,y)$ on $\partial\lambda$ as a product over
infinite ``angles'' placed at concave and convex corners of
$\lambda$. Each angle corresponds to $\mathsf{Y}(\tilde{q}_3^a
\tilde{q}_4^b x)$ ($(a,b)$ a concave corner) or
$\left(\mathsf{Y}(\tilde{q}_3^a \tilde{q}_4^b x)\right)^{-1}$ ($(a,b)$
a convex corner). The argument can be visualized using
Fig.~\ref{fig:1} in which the angles are drawn as thin lines with blue
color corresponding to $V_{-}$ operators and red to $V_{+}$ ones. It
is clear from the picture that all the lines except $\partial\lambda$
cancel. Hence we have
\begin{equation}
  \label{eq:100}
  \lim_{M\to \infty}  :\prod_{(a',b') \in \partial
    \lambda} \left( V_{s(a',b')} (\mu^{a'+k} q_3^{b' - k/2} x)
  \right):\, =\, : \prod_{(a,b) \in \mathrm{CC}(\lambda)}
  \mathsf{Y}(\tilde{q}_3^a \tilde{q}_4^b x) \prod_{(c,d) \in \mathrm{CV}(\lambda)}
  \mathsf{Y}(\tilde{q}_3^c \tilde{q}_4^d x):.
\end{equation}
Eq.~\eqref{eq:100} implies Proposition~\ref{prop-infty} since every
term in the $p$-expansion of $ \widehat{R}^{q_3,
  M}_{q_1,q_2}(x|v|p,\mu)$ has a well-defined limit written in terms
of $\mathsf{Y}$ operators.

Combining Eqs.~\eqref{eq:92}, \eqref{eq:93}, and~\eqref{eq:100}
together we get the statement of Theorem~\ref{thm-jacobi}.
\end{proof}

\begin{rem}
  Affine $qq$-characters commute with affine screening charges
  (integrals of $S_0^0 S_0(z)$ from Eq.~\eqref{eq:24} for $N=1$). It
  follows from the intertwining property of the operators $\Phi$ and
  $\Phi^{*}$ from which the screening charge is built and the
  $R$-matrices from which the spiral transfer matrix is built.
\end{rem}
\begin{rem}
  We will use the following abbreviated version of the
  diagram~\eqref{eq:75}:
  \begin{equation}
    \label{eq:119}
      \widehat{R}^{q_3, M}_{q_1,q_2}(x|v|p,\mu)    =     \quad  \includegraphics[valign=c]{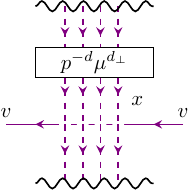}
  \end{equation}
\end{rem}
  
\subsection{Higher $\widehat{A}_0$ $qq$-characters}
\label{sec:higher-a_0-qq}
Higher $qq$-characters of $\widehat{A}_0$ type are just (ordered)
products of the operators~\eqref{eq:78} with different arguments
\begin{dfn}
  Let $x_i \in \mathbb{C}^{*}$, $i=1,\ldots,N$ be generic
  parameters. We call
\begin{equation}
  \label{eq:102}
  \mathsf{X}^{(N)}(\vec{x}|q_1,q_2, \tilde{q}_3, \tilde{q}_4, p) =  \prod_{i=1}^{
    \begin{smallmatrix}
      \to\\
      N
    \end{smallmatrix}
} \mathsf{X}(x_i|q_1,q_2, \tilde{q}_3, \tilde{q}_4, p),
\end{equation}
a higher $qq$-character of $\widehat{A}_0$ type.
\end{dfn}
\begin{prop} We have the following explicit expression
  \begin{multline}
    \label{eq:103}
    \mathsf{X}^{(N)}(\vec{x}|q_1,q_2, \tilde{q}_3, \tilde{q}_4, p) =
\left(\prod_{1 \leq i<j \leq N} \mathcal{S} \left(
      \frac{x_j}{x_i} \right) \right)    \sum_{\vec{\lambda}} p^{\sum_{i=1}^N |\lambda^{(i)}|}
    \prod_{i=1}^N z^{\mathrm{adj}}_{\lambda^{(i)}} ( \tilde{q}_3,
    \tilde{q}_4, q_1, q_2 )\times\\
    \times \prod_{1 \leq i<j \leq N} 
    z^{\mathrm{adj}}_{\lambda^{(i)},\lambda^{(j)}} \left(\left.
        \frac{x_j}{x_i} \right| \tilde{q}_3, \tilde{q}_4, q_1, q_2
    \right)\, : \prod_{i=1}^N \prod_{(a,b)\in \mathrm{CC}(\lambda^{(i)})}
    \mathsf{Y}(\tilde{q}_3^a \tilde{q}_4^b x_i) \prod_{(c,d)\in
      \mathrm{CV}(\lambda^{(i)})} \left(\mathsf{Y}(\tilde{q}_3^c
      \tilde{q}_4^d x_i)\right)^{-1}:
\end{multline}
where the sum is over $N$-tuples of Young diagrams $\vec{\lambda} =
(\lambda^{(1)}, \ldots, \lambda^{(N)})$,
\begin{equation}
  \label{eq:106}
  \mathcal{S}_{q_1,q_2,\tilde{q}_3, \tilde{q}_4}(x) = \exp \left[ - \sum_{n \geq 1} \frac{x^n}{n}
    \frac{(1-q_1^n)(1-q_2^n)}{(1-\tilde{q}_3^{-n})(1-\tilde{q}_4^{-n})}
  \right] = \frac{(x; \tilde{q}_3^{-1}, \tilde{q}_4^{-1})_{\infty}
    (q_1 q_2 x; \tilde{q}_3^{-1}, \tilde{q}_4^{-1})_{\infty}}{(q_1 x;
    \tilde{q}_3^{-1}, \tilde{q}_4^{-1})_{\infty} (q_2 x; \tilde{q}_3^{-1}, \tilde{q}_4^{-1})_{\infty}},
\end{equation}
and
\begin{multline}
  \label{eq:104}
  z^{\mathrm{adj}}_{\lambda,\mu} \left(\left.
    x \right| \tilde{q}_3, \tilde{q}_4, q_1, q_2
\right) = \\
= \left(\prod_{(a,b) \in \lambda} \psi_{q_1,q_2} \left( x\,
    \tilde{q}_3^{\mathrm{Leg}_{\mu}(a,b)+1} \tilde{q}_4^{-
      \mathrm{Arm}_{\lambda}(a,b)} \right) \right) \left( \prod_{(c,d)
    \in \mu} \psi_{q_1,q_2} \left( x\,
    \tilde{q}_3^{-\mathrm{Leg}_{\lambda}(c,d)} \tilde{q}_4^{\mathrm{Arm}_{\mu}(c,d)+1} \right) \right).
\end{multline}
\end{prop}
\begin{proof}
  We only need to find the factors that appear due to normal ordering
  of the $\mathsf{Y}$ operators belonging to different
  $\mathsf{X}(x_i|q_1,q_2, \tilde{q}_3, \tilde{q}_4, p)$. The sum over
  $\vec{\lambda}$ up to an overall factor was obtained e.g.\
  in\footnote{Of course, it can also be rederived directly by applying
    Wick's theorem to the product of $Y$'s.}~\cite{Kimura:2022spi}. To
  establish the proportionality factor it is enough to consider the
  term with all $\lambda^{(i)} = \varnothing$:
  \begin{equation}
    \label{eq:108}
    \mathsf{X}^{(N)}(\vec{x}|q_1,q_2, \tilde{q}_3, \tilde{q}_4, p) =
    \prod_{i=1}^{
      \begin{smallmatrix}
        \to\\
        N
      \end{smallmatrix}
    } \mathsf{Y}(x_i) + \mathcal{O}(p).
  \end{equation}
  Using the Wick's theorem we find
  \begin{equation}
    \label{eq:105}
    \mathsf{Y}(x) \mathsf{Y}(y) =  \, :\mathsf{Y}(x) \mathsf{Y}(y):
    \mathcal{S} \left( \frac{y}{x} \right),
  \end{equation}
  so that
  \begin{equation}
    \label{eq:109}
    \mathsf{X}^{(N)}(\vec{x}|q_1,q_2, \tilde{q}_3, \tilde{q}_4, p) =
\prod_{1 \leq i<j \leq N} \mathcal{S}_{q_1,q_2,\tilde{q}_3, \tilde{q}_4} \left(
      \frac{x_j}{x_i} \right)    \, :\prod_{i=1}^N \mathsf{Y}(x_i):\, + \mathcal{O}(p)
  \end{equation}
  which gives the desired formula~\eqref{eq:103}.
\end{proof}

Applying Theorem~\ref{thm-jacobi} to the higher
$qq$-character~\eqref{eq:103} we get
\begin{cor}[Noncommutative Jacobi identity for higher $qq$-characters]
  \begin{equation}
    \label{eq:110}
    \prod_{i=1}^{
      \begin{smallmatrix}
        \to\\
        N
      \end{smallmatrix}
} \widehat{R}^{q_3, \infty}_{q_1,q_2}(x_i|v|p,\mu) = \sum_{\vec{k} \in \mathbb{Z}^N}
    p^{\sum_{i=1}^N \frac{k_i^2}{2}} (-\tilde{v})^{-\sum_{i=1}^N k_i}
    \mathsf{X}^{(N)}\left(\tilde{q}_4^{- \vec{k}} \vec{x}|q_1,q_2,
      \tilde{q}_3, \tilde{q}_4, p \right)
    q_3^{- \sum_{i=1}^N k_i x_i \partial_{x_i}}
  \end{equation}
  where $\tilde{q}_3$, $\tilde{q}_4$ are defined by
  Eqs.~\eqref{eq:73},~\eqref{eq:97} and $\tilde{q}_4^{- \vec{k}}
  \vec{x} = (\tilde{q}_4^{- k_1} x_1, \ldots, \tilde{q}_4^{- k_N}
  x_N)$.
  \label{cor-higher-qq}
\end{cor}

It is instructive to draw a diagram corresponding to the product of
spiral transfer matrices in the l.h.s.\ of~\eqref{eq:110}. We use the
abbreviated notation~\eqref{eq:119}. The product of spiral transfer
matrices in the Jacobi identity for higher
$qq$-characters~\eqref{eq:110} is given by
\begin{equation}
  \label{eq:116}
\prod_{i=1}^{
      \begin{smallmatrix}
        \to\\
        N
      \end{smallmatrix}
    } \widehat{R}^{q_3, \infty}_{q_1,q_2}(x_i|v|p,\mu) = \includegraphics[valign=c]{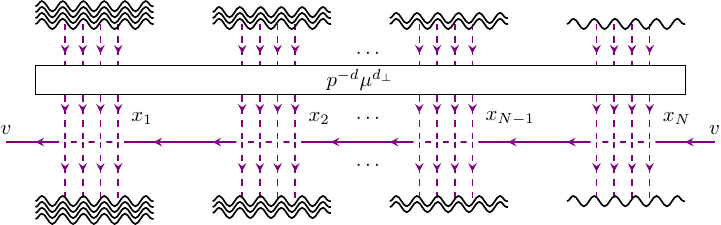}
\end{equation}
Let us remind that each ``coil'' in Eq.~\eqref{eq:116} is actually a
spiral winding an infinite number of times around the compactification
cylinder.

Finally we would like to reorganize the diagram~\eqref{eq:116}
similarly to what we did with Shiraishi functions in
Eq.~\eqref{eq:114}. Instead of composing the operators corresponding
to infinite coils with each coil centered around $x_i$ we first take
the the product of $N$ $R$-matrices at $x_i$ and then wrap $N$ dashed
ends around the cylinder together. We make the following
\begin{conj}
  Diagram
  \begin{equation}
    \label{eq:117}
    \includegraphics[valign=c]{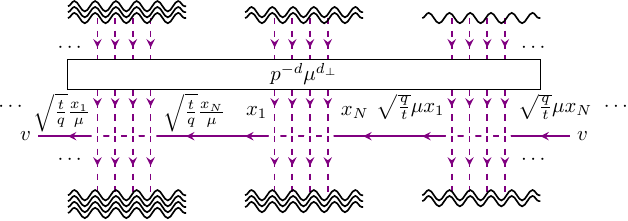}
  \end{equation}
  and~\eqref{eq:116} give identical operators.
\label{conj-op}
\end{conj}
A possible proof seems to require an infinite number of applications
of the Yang-Baxter equation.

\section{Elliptic deformations of the RS system}
\label{sec:comp-vs-spir}

In this section we introduce elliptic deformations of the setup of
sec.~\ref{sec:miura-r-matrix}. In sec.~\ref{sec:ellipt-deform-as} we
describe \emph{spiralization} of both the generating function of the
trigonometric RS Hamiltonians $\mathcal{O}_N(\vec{x}|v)$ and the
network of intertwiners corresponding to their eigenfunctions. The
former gives rise to the spiral transfer matrices considered in
sec.~\ref{sec:nonc-jacobi-ident} while the latter becomes the
Shiraishi functions as detailed in sec.~\ref{sec:affine-vert-oper}
(these relations require Conjectures~\ref{conj-redraw}
and~\ref{conj-op}). The spiralized setup has two extra parameters
compared to the trigonometric RS system, $p$ and $\kappa$.

We point out that there are at least two distinguished loci in the
$(\kappa, p)$ parameter plane. One of them corresponds to the
(trigonometric limit of the) Koroteev-Shakirov double elliptic system
(sec.~\ref{sec:korot-shak-syst}) and another to the eRS system
(sec.~\ref{sec:elliptic-rs-system}).

\subsection{Elliptic deformation as spiralization: motivation}
\label{sec:ellipt-deform-as}
Let us notice that compactification of a network of intertwiners and
$R$-matrices in the vertical and in the horizontal directions are the
same. Indeed, for any representations $A_1$, $A_2$, $B_1$, $B_2$ of
the algebra $\mathcal{A}$ let $\mathcal{T}: A_1 \otimes A_2 \to B_1
\otimes B_2$ be an intertwining operator. Since $\mathcal{T}$ is an
intertwiner it commutes with the simultaneous action of the grading
operators $d$ and $d_{\perp}$ on both tensor factors:
\begin{equation}
  \label{eq:88}
  (p^d \mu^{d\perp}\otimes p^d \mu^{d\perp}) \mathcal{T} = \mathcal{T} (p^d \mu^{d\perp}\otimes p^d \mu^{d\perp}).
\end{equation}
We therefore have
\begin{equation}
  \label{eq:112}
  \left( \mathcal{T} (p^d \mu^{d\perp}\otimes 1) \right)^K
  = \left( \mathcal{T} (p^d \mu^{d\perp}\otimes p^d
    \mu^{d\perp}) (1\otimes p^{-d} \mu^{-d\perp}) \right)^K = (p^d \mu^{d\perp}\otimes p^d
    \mu^{d\perp})^K  \left( \mathcal{T}  (1\otimes p^{-d} \mu^{-d\perp}) \right)^K.
\end{equation}
The overall grading operator $(p^d \mu^{d\perp}\otimes p^d
\mu^{d\perp})^K$ is usually immaterial and amounts to shifting the
spectral parameters of the representations. In pictorial language this
can be seen as a global translation a diagram.  In the cases we would
like to consider here $A_1$ will be a tensor product of vector
representations while $A_2$ a tensor product of horizontal Fock
representations. In the diagrammatic language the
identity~\eqref{eq:112} means that one can consider spiralization of a
network of intertwiners either in the vertical direction with vertical
dashed lines passing through the portals above and below the
picture or in the horizontal direction with horizontal solid lines
passing through portals to the left and right of the picture.

We illustrate the argument above by redrawing the
diagram~\eqref{eq:114} for the Shiraishi functions using the
identity~\eqref{eq:112}. In this case
\begin{align}
  \label{eq:111}
  A_1 &=
  (\mathcal{V}_{q_1}^{*})^{\otimes N},\\
  A_2 &= \bigotimes_{i=1}^N \mathcal{F}_{q_1,q_2}^{(1,0)}(\tilde{x}_{N-i+1}).
\end{align}
Applying Eq.~\eqref{eq:112} we get the following diagrammatic
expression for the Shiraishi functions:
\begin{equation}
  \label{eq:115}
  \widehat{G}_{q,t^{-1}}(\vec{s})  \widehat{f}_N(\vec{x}, p^{\frac{1}{N}} | \vec{s},
    \kappa^{\frac{1}{N}}|q,t)\stackrel{\mathrm{?}}{=}\quad  \includegraphics[valign=c]{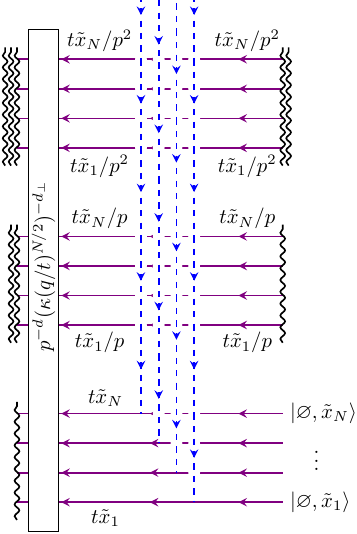}
\end{equation}
The identification~\eqref{eq:114} depended on
Conjecture~\ref{conj-shir} in the first place, hence the question
mark.

It is natural to act on the diagram~\eqref{eq:115} with an operator
down the blue dashed lines sticking out on the top of the picture. The
natural operator for this task is given by the spiralization of the
transfer matrix~\eqref{eq:46}. To this end we apply the observation
above to the spiral transfer matrix in the form Eq.~\eqref{eq:117}
with $q_1 \leftrightarrow q_3$. The result is
\begin{equation}
  \label{eq:118}
  \prod_{i=1}^{
      \begin{smallmatrix}
        \to\\
        N
      \end{smallmatrix}
    } \widehat{R}^{q_1, \infty}_{q_2,q_3}\left(x_i|v|p,\kappa q_3^{-N/2}\right) \stackrel{\mathrm{?}}{=} \quad  \includegraphics[valign=c]{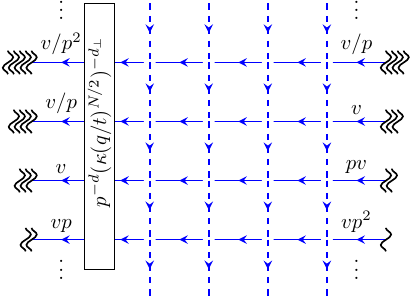}
\end{equation}
The question mark is again to indicate that the form~\eqref{eq:117}
that we have started with requires Conjecture~\ref{conj-op}.

Following the motivating arguments above we demonstrate how for two
particular values of $(p,\kappa)$ the spiralling operators and
prospective eigenfunctions simplify. 

\subsection{Trigonometric Koroteev-Shakirov system}
\label{sec:korot-shak-syst}
In this section we obtain trigonometric Koroteev-Shakirov Hamiltonians
from the product of spiral transfer matrices~\eqref{eq:110}. Recall
the definition of the trigonometric limit of the double elliptic
Koroteev-Shakirov Hamiltonians.
\begin{dfn}[\cite{Koroteev:2019gqi}]
  Let $q_2, q_3 \in \mathbb{C}^{*}$ and let $p$ be a formal
  parameter. The trigonometric Koroteev-Shakirov operator
  $\widehat{\mathcal{O}}^{\mathrm{KS}}_N(\vec{x}|v,q_1,q_2,p):(\mathcal{V}_{q_1}^{*})^{\otimes
    N} \to (\mathcal{V}_{q_1}^{*})^{\otimes N}$ is given by
  \begin{equation}
    \label{eq:113}
    \widehat{\mathcal{O}}^{\mathrm{KS}}_N(\vec{x}|v,q_2,q_3,p) =
    \sum_{\tilde{k} \in \mathbb{Z}^N} p^{\sum_{i=1}^N \frac{k_i^2}{2}}
    (-v)^{\sum_{i=1}^N k_i} \prod_{i<j} \frac{q_2^{-k_i} x_i -
      q_2^{-k_j} x_j}{x_i - x_j} q_3^{\sum_{i=1}^N k_i x_i \partial_{x_i}}
  \end{equation}
  Trigonometric Koroteev-Shakirov Hamiltonians
  $H^{\mathrm{tKS},(N)}_n$ are given by
  \begin{equation}
    \label{eq:121}
      \left[ \widehat{\mathcal{O}}^{\mathrm{KS}}_N\left(\vec{x}\left| v,q_1,q_2,p
      \right.\right)
    \right]_0^{-1}    \widehat{\mathcal{O}}^{\mathrm{KS}}_N(\vec{x}|v,q_2,q_3,p) 
 = \sum_{n \in \mathbb{Z}} H^{\mathrm{tKS},(N)}_n v^n,
  \end{equation}
where $[\ldots]_0$ denote the coefficient in front of $v^0$.
\end{dfn}
\begin{rem}
  Due to the quasiperiodicity of
  $\widehat{\mathcal{O}}^{\mathrm{KS}}_N(\vec{x}|v,q_2,q_3,p)$ there
  are only $N$ independent Hamiltonians $H^{\mathrm{tKS},(N)}_n$.
\end{rem}
\begin{rem}
  The coefficients of the Laurent expansion of
  \begin{equation}
    \label{eq:122}
\left( \widehat{\mathcal{O}}^{\mathrm{KS}}_N(\vec{x}|u,q_2,q_3,p) \right)^{-1}    \widehat{\mathcal{O}}^{\mathrm{KS}}_N(\vec{x}|v,q_2,q_3,p) 
  \end{equation}
  in $u$ and $v$ are linear combinations of $H^{\mathrm{tKS},(N)}_n$,
  so one can consider~\eqref{eq:122} instead of~\eqref{eq:121} as the
  generating function of the set of Hamiltonians
  $H^{\mathrm{tKS},(N)}_n$.
\end{rem}
\begin{conj}[\cite{Koroteev:2019gqi}]
  \begin{equation}
    \label{eq:123}
    [H^{\mathrm{tKS},(N)}_n, H^{\mathrm{tKS},(N)}_m] = 0 \qquad
    \text{for all } n, m.
  \end{equation}
\end{conj}

\begin{dfn}
  Introduce the vacuum-vacuum matrix element of the product of spiral
  transfer matrices:
  \begin{equation}
  \label{eq:107}
  \widehat{\mathcal{O}}_N^{\mathrm{spiral}}(\vec{x}|v,q_2,q_3,\mu,p) = \langle \varnothing, v |
  \prod_{i=1}^{
      \begin{smallmatrix}
        \to\\
        N
      \end{smallmatrix}
} \widehat{R}^{q_3, \infty}_{q_1,q_2}(x_i|v|p,\mu) |
\varnothing, v\rangle  
\end{equation}
where $\widehat{R}^{q_3, \infty}_{q_1,q_2}(x_i|v|p,\mu)$ is given by
Eq.~\eqref{eq:66}.
\end{dfn}

\begin{thm}
  We have the following identity
  \begin{equation}
    \label{eq:124}
    \prod_{1 \leq i<j \leq N}\left(1 - \frac{x_j}{x_i}\right) \widehat{\mathcal{O}}^{\mathrm{KS}}_N(\vec{x}|\tilde{v},q_2,q_3,p)
    = \widehat{\mathcal{O}}_N^{\mathrm{spiral}}\left(\vec{x}\left|v,q_2,q_3, \mu ,p\right.\right),
  \end{equation}
  where we set
  \begin{equation}
    \label{eq:127}
    \mu = \sqrt{\frac{q_1}{q_2}} 
  \end{equation}
   and $\tilde{v} = \frac{v}{q_3}$.
\end{thm}
\begin{proof}
  We use the noncommutative Jacobi identity~\eqref{eq:110} for higher
  $qq$-characters. Special value~\eqref{eq:127} corresponds to
  $(\tilde{q}_3, \tilde{q}_4) = (q_1^{-1},q_2^{-1})$. From the
  explicit form of the function $\psi_{q_1,q_2}(z)$ we have
  \begin{equation}
    \label{eq:125}
    z^{\mathrm{adj}}_{\lambda} ( q_1^{-1}, q_2^{-1}, q_1, q_2 ) =
    \delta_{\lambda , \varnothing},
  \end{equation}
  hence the sum over $\lambda^{(i)}$ in Eq.~\eqref{eq:103} reduces to
  the single term with $\lambda^{(i)} = \varnothing$.

  The function $\mathcal{S}_{q_1,q_2,\tilde{q}_3,
    \tilde{q}_4}(z)$ from Eq.~\eqref{eq:106} also simplifies for the
  choice of parameters~\eqref{eq:127}:
  \begin{equation}
    \label{eq:126}
    \mathcal{S}_{q_1,q_2, q_1^{-1}, q_2^{-1}}(x) = 1 - x. 
  \end{equation}
  The vacuum-vacuum matrix element of the normal ordered product of
  $\mathsf{Y}$ operators in Eq.~\eqref{eq:103} gives the
  identity. Thus, for $\mu = \sqrt{\frac{q_1}{q_2}}$ the
  noncommutative Jacobi identity~\eqref{eq:110} reduces to
  \begin{equation}
    \label{eq:128}
    \widehat{\mathcal{O}}_N^{\mathrm{spiral}}\left(\vec{x}\left|v,q_2,q_3,
        \left( \frac{q_1}{q_2} \right)^{\frac{1}{2}}    ,p\right.\right) = \sum_{\vec{k} \in \mathbb{Z}^N}
    p^{\sum_{i=1}^N \frac{k_i^2}{2}} (-\tilde{v})^{-\sum_{i=1}^N k_i}
    \prod_{i<j} \left(1 - \frac{q_2^{k_j} x_j}{q_2^{k_i} x_i} \right)
    q_3^{- \sum_{i=1}^N k_i x_i \partial_{x_i}}.
  \end{equation}
  Comparing with Eq.~\eqref{eq:113} one obtains the result of the
  theorem.
\end{proof}
\begin{rem}
  From the identity~\eqref{eq:124} we find that the system of KS
  Hamiltonians can be obtained from a system of spiralling
  $R$-matrices of algebra $\mathcal{A}$. We hope that a proof of
  commutativity~\eqref{eq:123} may be obtained by using 
  manipulations similar to those in sec.~\ref{sec:miura-r-matrix-1}.
\end{rem}

\subsection{Elliptic RS system}
\label{sec:elliptic-rs-system}
For certain choices of $p$ the spirals in
Eqs.~\eqref{eq:115},~\eqref{eq:118} turn into circles and the
corresponding system of intertwiners simplifies. Indeed, if the
spectral parameter of the representations on two neighbouring turns of
the spiral is the same one can actually take a trace over a single
winding.

However, the values at which this happens are \emph{different} for the
spiralized wavefunction Eq.~\eqref{eq:115} and for the spiralized
operator Eq.~\eqref{eq:118}. Indeed, since the spectral parameter $v$
is the same along the horizontal line, the operator~\eqref{eq:118}
should degenerate into a circle for $p=1$. We conjecture that the
degeneration actually happens in the limit $p \to 1$ (possibly after
rescaling):
\begin{conj}
  \begin{multline}
  \label{eq:129}
  \lim_{p \to 1}    \prod_{i=1}^{
      \begin{smallmatrix}
        \to\\
        N
      \end{smallmatrix}
    } \widehat{R}^{q_1, \infty}_{q_2,q_3}\left(x_i|v|p,\kappa
      q_3^{-N/2}\right) =
    \quad \includegraphics[valign=c]{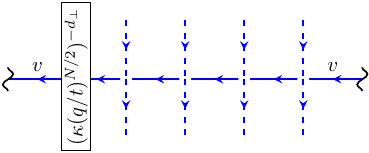}
    =\\
    = \Tr_{\mathcal{F}_{q_2,q_3}^{(1,0)}(v)} \left( \left(\kappa
        q_3^{-N/2}\right)^{-d_{\perp}} \prod_{i=1}^{
      \begin{smallmatrix}
        \to\\
        N
      \end{smallmatrix}} R_{q_2,q_3}^{q_1}(x_i|v) \right)
\end{multline}
\end{conj}
Next we prove that the trace in the r.h.s.\ of Eq.~\eqref{eq:129} is
in fact the generating function of eRS Hamiltonians.
\begin{dfn}
  We call
  \begin{equation}
    \label{eq:131}
    \mathcal{O}^{\mathrm{ell}}_N(\vec{x}|v,\mu) = \Tr_{\mathcal{F}_{q_2,q_3}^{(1,0)}(v)} \left( \mu^{-d_{\perp}} \check{\mathcal{R}}|_{(\mathcal{V}^{*}_{q_1})^{\otimes N} \otimes
        \mathcal{F}_{q_2,q_3}^{(1,0)}(v)} \right) = \Tr_{\mathcal{F}_{q_2,q_3}^{(1,0)}(v)} \left( \mu^{-d_{\perp}} \prod_{i=1}^{
      \begin{smallmatrix}
        \to\\
        N
      \end{smallmatrix}} R_{q_2,q_3}^{q_1}(x_i|v) \right)
\end{equation}
the elliptic transfer matrix.
\end{dfn}
\begin{dfn}
  Elliptic Ruijsenaars-Schneider Hamiltonians are given by
  \begin{equation}
    \label{eq:132}
    H_k^{\mathrm{eRS},(N)}(\vec{x},q_1,q_2,\mu) = q_2^{-\frac{k(k-1)}{2}}\sum_{
      \begin{smallmatrix}
        I \subseteq \{1,\ldots, N \}\\
        |I|=k
      \end{smallmatrix}
    } \prod_{i \in I} \prod_{j \not \in I} \frac{\theta_{\mu} \left(
        q_2^{-1} \frac{x_i}{x_j} \right)}{\theta_{\mu} \left(
        \frac{x_i}{x_j} \right)} \prod_{i \in I} q_1^{x_i \partial_{x_i}}, \qquad k=0,\ldots,N,
  \end{equation}
  where
  \begin{equation}
    \label{eq:133}
    \theta_{\mu}(x) = (\mu;\mu)_{\infty} (x;\mu)_{\infty} \left(
      \frac{\mu}{x} ; \mu \right)_{\infty}
  \end{equation}
  is the Jacobi theta function.
\end{dfn}

The next statement is very similar to Theorem~\ref{thm-tRS-ham}.
\begin{thm}
  \begin{equation}
    \label{eq:130}
\left( F_{q_1,q_2,\mu}^{\mathrm{ell}}(\vec{x})
\right)^{-1}\mathcal{O}^{\mathrm{ell}}_N(\vec{x}|v,\mu)F_{q_1,q_2,\mu}^{\mathrm{ell}}(\vec{x})
= \sum_{k=0}^N (-v/q_1)^{N-k} H_k^{\mathrm{eRS},(N)}(\vec{x},q_1,q_2,\mu),
\end{equation}
where
\begin{equation}
  \label{eq:134}
  F_{q_1,q_2,\mu}^{\mathrm{ell}}(\vec{x}) = \prod_{k<l}^N \left[ x_k^{- \frac{\ln
        q_2}{\ln q_1}} \prod_{n \geq 0}\frac{\theta_{\mu} \left( q_1^n\frac{x_l}{x_k}\right)}{\theta_{\mu}\left( q_2^{-1} \frac{x_l}{x_k} \right)} \right]
\end{equation}
\label{thm-ers}
\end{thm}
\begin{proof}
We use the standard Heisenberg algebra identities to take the trace
over the Fock space. We have the analogue of the commutation relations~\eqref{eq:50}:
\begin{equation}
  \label{eq:135}
  \Tr_{\mathcal{F}_{q_2,q_3}^{(1,0)}(v)} \left( \mu^{-d_{\perp}}
    V^{+}_{q_2,q_3}(x) V^{-}_{q_2,q_3}(x) \right) = \psi_{q_2,q_3,\mu}^{\mathrm{ell}} \left( \frac{y}{x} \right),
\end{equation}
where
\begin{equation}
  \label{eq:136}
  \psi_{q_2,q_3,\mu}^{\mathrm{ell}}(z) = \frac{\theta_{\mu}(q_i z)\theta_{\mu}(q_j z)}{\theta_{\mu}(z)\theta_{\mu}(q_i
    q_j z)}.
\end{equation}
The rest of the proof is exactly parallel to that of
Theorem~\ref{thm-tRS-ham}.
\end{proof}
\begin{thm}
  \begin{equation}
    \label{eq:151}
    [H_k^{\mathrm{eRS},(N)}(\vec{x},q_1,q_2,\mu),
    H_l^{\mathrm{eRS},(N)}(\vec{x},q_1,q_2,\mu)] = 0, \qquad \text{for
    all} \quad  k,l = 1,\ldots, N.
\end{equation}
\label{thm-ers-comm}
\end{thm}
\begin{proof}
  One can use an argument very similar to the proof of
  Theorem~\ref{thm-O-comm} except one needs to move the $R$-matrices
  along the circle. It is crucial that the horizontal lines are
  circles and not spirals so that $\check{\mathcal{R}}$ and
  $\check{\mathcal{R}}$ cancel after passing around the loop in the
  opposite directions.
\end{proof}

Finally, we consider the case when the spiralized
wavefunction~\eqref{eq:115} reduces to a trace over $N$ Fock
spaces. As can be seen from the diagram~\eqref{eq:115} this should
happen when $p = t$.
\begin{conj}
  \begin{equation}
  \label{eq:137}
 \lim_{p \to t} \widehat{G}_{q,t^{-1}}(\vec{s})  \widehat{f}_N(\vec{x}, p^{\frac{1}{N}} | \vec{s},
    \kappa^{\frac{1}{N}}|q,t)=\quad  \includegraphics[valign=c]{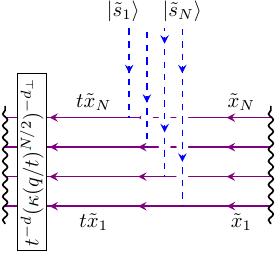}
\end{equation}
\end{conj}
The setup in the r.h.s.\ of Eq.~\eqref{eq:137} in fact reproduces the
results of~\cite{Fukuda:2020czf} where Shiraishi functions for $\kappa
= t^{-1}$ were built using a system of Fock space intertwiners of the
algebra $\mathcal{A}$. This should fit with our limit $p \to t$ after
making a spectral duality transformation on Shiraishi functions which
exchanges $\kappa \leftrightarrow p$ and $\tilde{s}_i \leftrightarrow
\tilde{x}$. On can see the equivalence of~\cite{Fukuda:2020czf} with
the diagram in Eq.~\eqref{eq:137} either by a direct calculation or
using the Higgsing technique as explained
in~\cite{Zenkevich:2020ufs}. This will be done elsewhere.

Notice that elliptic partition function~\eqref{eq:137} is \emph{not}
an eigenfuction of the eRS Hamiltonians~\eqref{eq:132}. The $R$-matrix
manipulation used to prove the trigonometric limit of this statement
fails when both the diagrams for the operator and the eigenfunction
are compactified as \emph{circles,} because the shift operators
$p^{-d}$ are different in~\eqref{eq:137} and~\eqref{eq:132}.

\section{Conclusions}
\label{sec:conclusions}

We have considered spiralling systems of intertwiners of the quantum
toroidal algebra $\mathcal{A}$ and have found them useful in the
theory of elliptic integrable systems. We have found a neat $R$-matrix
formalism for the tRS and eRS models, given a new interpretation of
the important Shiraishi functions as matrix elements of a system of
intertwiners. We have proven the noncommutative Jacobi identity for
general equivariant parameters and have connected it with
Koroteev-Shakirov Hamiltonians.

Of course it is important to prove the ``reordering''
Conjectures~\ref{conj-redraw},~\ref{conj-op}. This would allow for a
very explicit construction of all the landscape of elliptic systems
and dualities between them using spiralling systems of intertwiners of
the quantum toroidal algebra.

Let us mention some possible further applications of our approach. The
natural generalization of elliptic integrable systems are double
elliptic (Dell) systems~\cite{Awata:2019isq}. We hope that spiralling
branes could be used to understand these models better as well.

It should be possible to prove the noncommutative Jacobi identity for
$\widehat{A}_N$ type $qq$-characters with $N > 0$ using the techniques
of sec.~\ref{sec:nonc-jacobi-ident}: one needs to make the same
argument for $(N+1)$ horizontal lines intersected by the dashed spiral
of the same color.

Let us also comment on the relation between this work
and~\cite{Mironov:2021sfo}. There polynomial eigenfunctions of the
dual eRS system were found and related to polynomial eigenfunctions of
the trigonometric Koroteev-Shakirov Hamiltonians. The polynomial
eigenfunctions are particular cases of Shiraishi functions, so it
should be possible to understand them using the spiral system from
sec.~\ref{sec:affine-vert-oper}. Indeed, it turns out that for
particular values of the spectral parameters the tensor products of
representations on which the system of intertwiners acts become
reducible which explains the polynomialty of the answer. We hope that
our algebraic approach may allow to understand the dualities between
various elliptic systems featuring in~\cite{Mironov:2021sfo} more
explicitly. We plan to return to this topic in the future.

Finally, we comment on the relations to other algebraic approaches to
gauge theories and integrable systems. There is a large body of work
on the vertex operator formalism for gauge origami partition
functions~\cite{Kimura:2023bxy}. Spiralling branes should arise there
as well. The advantage of the formalism that we are pursuing here is
that the symmetry algebra $\mathcal{A}$ is built into the construction
from the start, whereas in more general operator setups (see
e.g.~\cite{Kimura:2015rgi}) the algebra is defined implicitly as the
commutant of a set of screening charges on a case-by-case basis.

It would also be interesting to understand the connection of our
results with the recent work~\cite{Tamagni:2024iiy} in which Yangian
algebras are derived from quantized Coulomb branches of supersymmetric
$3d$ gauge theories.

\section*{Acknowledgements}
\label{sec:acknowledgements}

Author's work is partially supported by the European Research Council
under the European Union’s Horizon 2020 research and innovation
programme under grant agreement No~948885.

The author would like to thank M.~Aganagic, S.~Tamagni, N.~Haouzi,
N.~Nekrasov, A.~Grekov, M.~Bershtein, M.~Semenyakin and A.~Mironov for
discussions.

\appendix

\section{Quantum toroidal algebra $\mathcal{A}$}
\label{sec:quant-toro-algebra}
We write $\mathcal{A}=U_{q_1,q_2}
(\widehat{\widehat{\mathfrak{gl}}}_1)$ for the quantum toroidal
algebra of type $\mathfrak{gl}_1$ (also known as the DIM algebra)
throughout the paper.

\begin{dfn} 
  Let $q_1,q_2 \in \mathbb{C}^{*}$ be parameters and let $q_3 =
  \frac{1}{q_1 q_2}$. The algebra $\mathcal{A}$ is generated by
  $x^{\pm}_n$ ($n \in \mathbb{Z}$), $\psi^{\pm}_{\pm n}$ ($n \in
  \mathbb{Z}_{\geq 0}$) and the central element $\gamma$. The
  generating functions
  \begin{equation}
    \label{eq:143}
    x^{\pm}(z) = \sum_{n \in \mathbb{Z}} x^{\pm}_n z^{-n}, \qquad
    \psi^{\pm}(z) = \sum_{n \geq 0} \psi^{\pm}_{\pm n} z^{\mp n}
  \end{equation}
  satisfy the relations
  \begin{gather}
    \label{eq:76}
      [\psi^{\pm}(z), \psi^{\pm}(w)] = 0,\qquad \psi^{+}(z) \psi^{-}(w) =
  \frac{g\left(\gamma \frac{w}{z}\right)}{g\left(\gamma^{-1}
      \frac{w}{z}\right)}
  \psi^{-}(w) \psi^{+}(z),\\
  \psi^{+}(z) x^{\pm}(w) = g\left(\gamma^{\mp \frac{1}{2}}
    \frac{w}{z}\right)^{\mp 1} x^{\pm}(w) \psi^{+}(z), \qquad
  \psi^{-}(z) x^{\pm}(w) = g\left(\gamma^{\mp \frac{1}{2}}
    \frac{w}{z}\right)^{\pm 1} x^{\pm}(w) \psi^{-}(z),\\
  [x^{+}(z), x^{-}(w)] = \frac{1}{G^{-}(1)} \left( \delta \left(
      \gamma^{-1} \frac{z}{w} \right) \psi^{+} \left(
      \gamma^{\frac{1}{2}} w \right) - \delta \left( \gamma
      \frac{z}{w} \right) \psi^{-} \left(
      \gamma^{-\frac{1}{2}} w \right) \right),\\
  G^{\mp}\left( \frac{z}{w} \right) x^{\pm}(z) x^{\pm}(w) =
  G^{\pm}\left( \frac{z}{w} \right) x^{\pm}(w) x^{\pm}(z),\\
  [x^{\pm}_n,[x^{\pm}_{n-1},x^{\pm}_{n+1}]]= 0 \label{eq:77}
\end{gather}
where $\delta(z) = \sum_{n \in \mathbb{Z}} z^n$, $G^{\pm}(z) =
\prod_{i=1}^{3} (1 - q_i^{\pm 1}z)$ and $g(z) = \frac{G^{+}(z)}{G^{-}(z)}$.
\end{dfn}
The algebra $\mathcal{A}$ is manifestly symmetric in the parameters
$q_i$. It also has a large $\widetilde{SL(2,\mathbb{Z})}$ automorphism
group, where tilde means the universal cover. $\mathcal{A}$ is doubly
graded with two gradings $d$ and $d_{\perp}$ acting as
\begin{align}
  \label{eq:71}
  [d, x^{\pm}_n] &= \pm x^{\pm}_n,\\
  [d, \psi^{\pm}_n] &= 0,\\
  [d_{\perp}, x^{\pm}_n] &= n x^{\pm}_n,\\
  [d_{\perp}, \psi^{\pm}_n] &= n \psi^{\pm}_n.\label{eq:72}
\end{align}

\subsection{Coproduct(s)}
\label{sec:coproducts}
Algebra $A$ has an infinite number of non-isomorphic coproducts. We
will use the following one:
\begin{align}
  \label{eq:141}
    \Delta(x^{+}(z)) &= x^{+}(z) \otimes 1 + \psi^{-}\left(
    \gamma_{(1)}^{\frac{1}{2}}z \right) \otimes x^{+} \left(
    \gamma_{(1)} z \right),\\
  \Delta(x^{-}(z)) &= x^{-}\left(\gamma_{(2)} z\right) \otimes
  \psi^{+}\left( \gamma_{(2)}^{\frac{1}{2}}z \right) + 1 \otimes
  x^{-}(z),\\
  \Delta(\psi^{\pm}(z)) &= \psi^{\pm}\left(
    \gamma_{(2)}^{\pm \frac{1}{2}}z \right) \otimes \psi^{\pm}\left(
    \gamma_{(1)}^{\mp \frac{1}{2}}z \right).
\end{align}

\subsection{Universal $R$-matrix}
\label{sec:universal-r-matrix-1}
\begin{dfn}
  Universal $R$-matrix is an element $\mathcal{R} \in \mathcal{A}
  \otimes \mathcal{A}$ such that
  \begin{equation}
    \label{eq:142}
    \mathcal{R} \Delta(g) = \Delta^{\mathrm{op}}(g) \mathcal{R}
  \end{equation}
  for any element $g \in \mathcal{A}$, where $\Delta^{\mathrm{op}}$
  denotes the opposite coproduct.
\end{dfn}
The universal $R$-matrix is given by~\cite{FJMM-R}
\begin{equation}
  \label{eq:4}
  \check{\mathcal{R}} = P\mathcal{R} 
\end{equation}
where $P$ is the permutation operator exchanging the factors in the
tensor product and
\begin{equation}
  \label{eq:7}
  \mathcal{R} = e^{-c \otimes d - d \otimes c - c_{\perp} \otimes
    d_{\perp} - d_{\perp} \otimes c_{\perp}} \exp \left[ -\sum_{n \geq 1}
    \frac{n}{\kappa_n} \gamma^{-\frac{n}{2}} H^{+}_n \otimes
      \gamma^{\frac{n}{2}} H^{-}_{-n}\right] \left( 1 + \kappa_1 \sum_{n\in \mathbb{Z}}
    x^{+}_n \otimes x^{-}_{-n} + \ldots \right)
\end{equation}
where $\kappa_n = \prod_{i=1}^3(1-q_i^n)$, $c = \ln
\psi_0^{+}/\psi_0^{-}$, $c_{\perp} = \ln \gamma$ and the omitted terms
contain higher powers of the currents $x^{+}_n$ and $x^{-}_{-n}$.
\subsection{Horizontal Fock representation}
\label{sec:horiz-fock-repr}
\begin{dfn}
  Let $a^{(q_i, q_j)}_n$ with $i\neq j \in \{1,2,3\}$, $n \in
  \mathbb{Z}\backslash \{0\}$ be the Heisenberg generators satisfying
  the commutation relations
\begin{equation}
  \label{eq:44}
  [a^{(q_i, q_j)}_n, a^{(q_i, q_j)}_m] = n \frac{1 - q_i^{|n|}}{1 - q_j^{-|n|}} \delta_{n+m,0}.
\end{equation}
Let $\mathcal{F}_{q_i,q_j}^{(1,0)}(u)$ be the Fock space generated by
the action of all $a_{-n}^{(q_i, q_j)}$ with $n>0$ from the vacuum
vector $|\varnothing,u\rangle$.

The action of $\mathcal{A}$ on $\mathcal{F}_{q_i,q_j}^{(1,0)}(u)$ is
given by
\begin{align}
  \label{eq:145}
  \rho_{\mathcal{F}_{q_i,q_j}^{(1,0)}(u)}(x^{+}(z)) &= \frac{u}{(1-q_i^{-1})
  (1-q_j^{-1})} \exp \left[ \sum_{n \geq 1} \frac{z^n}{n} (1-q_j^n)
    a_{-n}^{(q_i, q_j)} \right] \exp \left[ - \sum_{n
      \geq 1} \frac{z^{-n}}{n} (1-q_j^{-n}) a_n^{(q_i, q_j)} \right],\\
  \rho_{\mathcal{F}_{q_i,q_j}^{(1,0)}(u)}(x^{-}(z)) &= \frac{u^{-1}}{(1-q_i)
  (1-q_j)} \exp \left[ - \sum_{n \geq 1} \frac{z^n}{n}
    (1-q_j^n) \left( q_i q_j \right)^{-\frac{n}{2}} a_{-n}^{(q_i,
      q_j)} \right]\times \notag\\
  &\phantom{=} \times
  \exp \left[ - \sum_{n \geq 1} \frac{z^{-n}}{n} (1-q_j^{-n}) \left(
      q_i q_j
    \right)^{-\frac{n}{2}} a_n^{(q_i, q_j)} \right],\notag\\
  \rho_{\mathcal{F}_{q_i,q_j}^{(1,0)}(u)}(\psi^{+}(z)) &= \exp \left[
    - \sum_{n\geq 1} \frac{z^{-n}}{n} (1-q_j^{-n}) \left( 1 - \left(
        q_i q_j \right)^{-n} \right) \left( q_i q_j
    \right)^{\frac{n}{4}} a_n^{(q_i, q_j)}
  \right],\\
  \rho_{\mathcal{F}_{q_i,q_j}^{(1,0)}(u)}(\psi^{-}(z)) &= \exp \left[
    \sum_{n\geq 1} \frac{z^n}{n} (1-q_j^n) \left( 1 - \left(
        q_i q_j \right)^{-n} \right) \left( q_i q_j
    \right)^{\frac{n}{4}} a_{-n}^{(q_i, q_j)}
  \right],\\
  \rho_{\mathcal{F}_{q_i,q_j}^{(1,0)}(u)}(\gamma) &= (q_i
  q_j)^{-\frac{1}{2}}.
\end{align}
\end{dfn}

It will be useful for us to add two generators $P$ and $Q$ with $P$
playing the role of $a_0^{(q_i, q_j)}$ and $Q$ shifting the spectral
parameter $u$:
\begin{gather}
  \label{eq:144}
P| \varnothing,u\rangle = \ln u | \varnothing,u\rangle,  \\
  [P,a^{(q_i, q_j)}_n] = 0, \qquad   [Q,a^{(q_i, q_j)}_n] = 0, \qquad [P,Q]=1
\end{gather}
for $n \neq 0$. The grading operators act on the Fock representations
as follows:
\begin{equation}
  \label{eq:149}
  [d_{\perp}, a^{q_i, q_j}_n] = n a^{q_i, q_j}_n, \qquad d = Q.
\end{equation}

\subsection{Vector representations}
\label{sec:vect-repr}
One needs to be a bit careful in defining the vector representation
since there are slightly different objects called by the same
name.

\begin{dfn}
  Let $\mathcal{V}_{q_1}^{*}$ be the representation of $\mathcal{A}$
  on the space $\mathbb{C}((x))$ of functions of a single variable $x
  \in \mathbb{C}^*$ in which the generating currents of
  the algebra act as follows:
  \begin{align}
    \rho_{\mathcal{V}_{q_1}^{*}} (x^{+}(z)) f(x) &=
    -\frac{1}{1-q_1^{-1}} \delta \left( \frac{x}{q_1 z} \right)
    q_1^{-x \partial_x}f(x)
    , \notag\\
    \rho_{\mathcal{V}_{q_1}^{*}}(x^{-}(z))f(x) &= - \frac{1}{1-q_1}
    \delta \left( \frac{x}{z} \right)
    q_1^{x \partial_x}f(x), \notag\\
    \rho_{\mathcal{V}_{q_1}^{*}}(\psi^{+}(z))f(x) &= \frac{\left( 1 -
        q_3 \frac{x}{z}\right) \left( 1 - q_2
        \frac{x}{z}\right)}{\left( 1 - \frac{x}{z} \right) \left( 1 -
        \frac{1}{q_1} \frac{x}{z} \right)}f(x)
    ,\label{eq:6}\\
    \rho_{\mathcal{V}_{q_1}^{*}}(\psi^{-}(z))f(x) &= \frac{\left( 1 -
        \frac{1}{q_3} \frac{z}{x}\right) \left( 1 - \frac{1}{q_2}
        \frac{z}{x}\right)}{\left( 1 - \frac{z}{x} \right) \left( 1 -
        q_1 \frac{z}{x} \right)} f(x). \notag
  \end{align}
  Representations $\mathcal{V}_{q_2}^{*}$ and $\mathcal{V}_{q_2}^{*}$
  are obtained by the corresponding permutation of the parameters
  $q_i$.
\end{dfn}
\begin{rem}
  We put a star on $\mathcal{V}_{q_i}^{*}$ to conform with the
  notation in the literature, where by the vector representation one
  usually means (a subrepresentation of) $\mathcal{V}_{q_i}$.
\end{rem}
Consider a basis $\{ f_w(x) | w \in \mathbb{C}*\}$ in
$\mathcal{V}_{q_i}^{*}$ spanned by the delta functions $f_w(x) =
\delta \left( \frac{x}{w} \right) = \sum_{n \in \mathbb{Z}} \left(
  \frac{x}{w} \right)^n$. One can view the basis functions $f_w(x)$ as
matrix elements $\langle x | w\rangle$ on which the algebra
$\mathcal{A}$ acts by $q_i$-difference operators in $x$. The
formulas~\eqref{eq:6} are then valid with ``bra states'' $\langle x|$
in place of $f(x)$.

For completeness we also write out the dual representation of
$\mathcal{V}_{q_1}^{*}$.
\begin{cor}
  $\mathcal{V}_{q_1}$ is the representation of $\mathcal{A}$ on the
  space $\{|w\rangle\, |\, w \in \mathbb{C}^* \}$ in
  which the generators act as follows\footnote{The
    formulas~\eqref{eq:33} are obtained from Eq.~(73)
    of~\cite{Zenkevich:2018fzl} by inverting the signs of the currents
    $x^{\pm}_{\mathrm{here}}(z) =
    -x^{\pm}_{\cite{Zenkevich:2018fzl}}(z) $.}:
\begin{align}
  \rho_{\mathcal{V}_{q_1}} (x^{+}(z)) |w \rangle &=
  -\frac{1}{1-q_1^{-1}} \delta \left( \frac{w}{z} \right) | q_1 w
  \rangle = | w \rangle \left( -\frac{1}{1-q_1^{-1}} \delta \left(
      \frac{w}{q_1 z} \right) \overleftarrow{q_1^{w \partial_w}}
  \right)
  , \notag\\
  \rho_{\mathcal{V}_{q_1}}(x^{-}(z))|w \rangle &= - \frac{1}{1-q_1}
  \delta \left( \frac{w}{q_1 z} \right) \left| \frac{w}{q_1}
  \right\rangle= | w \rangle \left( -\frac{1}{1-q_1} \delta \left(
      \frac{w}{z}
    \right) \overleftarrow{q_1^{-w \partial_w}} \right), \notag\\
  \rho_{\mathcal{V}_{q_1}}(\psi^{+}(z))|w \rangle &= \frac{\left( 1 -
      q_3 \frac{w}{z}\right) \left( 1 - q_2 \frac{w}{z}\right)}{\left(
      1 - \frac{w}{z} \right) \left( 1 - \frac{1}{q_1} \frac{w}{z}
    \right)}|w \rangle
  ,\label{eq:33}\\
  \rho_{\mathcal{V}_{q_1}}(\psi^{-}(z))|w \rangle &= \frac{\left( 1 -
      \frac{1}{q_3} \frac{z}{w}\right) \left( 1 - \frac{1}{q_2}
      \frac{z}{w}\right)}{\left( 1 - \frac{z}{w} \right) \left( 1 -
      q_1 \frac{z}{w} \right)} |w \rangle, \notag\\
  \rho_{\mathcal{V}_{q_1}}(\gamma)& = 1. \notag
\end{align}
  The representations $\mathcal{V}_{q_2}$ and
  $\mathcal{V}_{q_2}$ are obtained by the corresponding
  permutation of the parameters $q_i$.
\end{cor}
The grading operators $p^d$ and $\mu^{d_{\perp}}$ act on the vector
representation $\mathcal{V}_{q_i}^{*}$ as follows
\begin{align}
  \label{eq:69}
  p^d f(x) &= x^{\frac{\ln p}{\ln q_i}} f(x),\\
  \mu^{d_{\perp}} f(x) &= \mu^{x \partial_x} f(x) = f(\mu x). \label{eq:70}
\end{align}

Since $\mathcal{A}$ acts on $\mathcal{V}_{q_i}^{*}$ by
$q_i$-difference operators one can consider a subrepresentation of
$\mathcal{V}_{q_i}^{*}$ spanned by the functions supported on $x =
q_i^n w$, $n \in \mathbb{Z}$.
\begin{dfn}
  Let $w \in \mathbb{C}^*$ be a spectral parameter. Let
  $\mathcal{V}_{q_1}^{*}(w) \subset \mathcal{V}_{q_1}^{*}$ be the
  subrepresentation obtained by restricting the action~\eqref{eq:6} to
  $f_{q_1^n w}(x) = \delta \left( \frac{x}{q_1^n w} \right)$, $n \in
  \mathbb{Z}$.

  Let $\mathcal{V}_{q_1}(w) \subset \mathcal{V}_{q_1}$ be the
  subrepresentation obtained by restricting the action~\eqref{eq:33}
  to the states $|q_i^nw\rangle$, $n \in \mathbb{Z}$.

  Similarly for $\mathcal{V}^{*}_{q_2}(w)$,  $\mathcal{V}^{*}_{q_3}(w)$,
  $\mathcal{V}_{q_2}(w)$,  $\mathcal{V}_{q_3}(w)$.
\end{dfn}
\begin{rem}
  Representations $\mathcal{V}_{q_i}(w)$ and $\mathcal{V}_{q_i}(q_iw)$
  are identical.
\end{rem}
\begin{rem}
  We have $\mu^{d_{\perp}}: \mathcal{V}_{q_i}(w) \to
  \mathcal{V}_{q_i}(\mu w)$, so for for generic $\mu$ the grading
  operator is not an operator in $\mathcal{V}_{q_i}(w)$. However, due
  to the previous remark, for $\mu = q_i^n$ with $n \in \mathbb{Z}$ it
  is.
\end{rem}

Informally, one should think of $\mathcal{V}_{q_i}$ as given by
some kind of ``direct integral'' of $\mathcal{V}_{q_i}(w)$ over
$\ln w \in T^2 = \mathbb{C}/(\mathbb{Z} + \ln q_i \mathbb{Z})$.

\end{document}